\documentclass{article} 
\usepackage{iclr2027_conference,times}

\usepackage{amsmath,amsfonts,bm}

\def\eqref#1{equation~\ref{#1}}

\def\1{\bm{1}}

\DeclareMathAlphabet{\mathsfit}{\encodingdefault}{\sfdefault}{m}{sl}
\SetMathAlphabet{\mathsfit}{bold}{\encodingdefault}{\sfdefault}{bx}{n}

\iclrfinalcopy

\usepackage{hyperref}
\usepackage{url}
\usepackage{booktabs}       
\usepackage{amsfonts}       
\usepackage{nicefrac}       
\usepackage{microtype}      
\usepackage{xcolor}         
\usepackage{graphicx} 
\usepackage{amsmath}
\usepackage{subcaption}
\usepackage{placeins}
\usepackage{amsthm}
\usepackage{wrapfig}
\usepackage{algorithm}
\usepackage{algorithmic}
\usepackage{comment}
\newtheorem{proposition}{Proposition}[section]
\newtheorem{lemma}[proposition]{Lemma}
\usepackage{multirow}

\title{CORE-BREW: LLR-Based Soft Decoding for Robust Multi-Bit LLM Watermarking}

\author{
    Joeun Kim \\
    Department of AI \\
    DGIST \\
    Daegu, Republic of Korea \\
    \texttt{jowithu@dgist.ac.kr} \\
    \And
    HoEun Kim \\
    Department of AI \\
    DGIST \\
    Daegu, Republic of Korea \\
    \texttt{hoeun.kim@dgist.ac.kr} \\
    \And
    Young-Sik Kim\thanks{Corresponding author.} \\
    Department of AI and Department of EECS \\
    DGIST \\
    Daegu, Republic of Korea \\
    \texttt{ysk@dgist.ac.kr}
}

\begin{document}

\maketitle
\lhead{Preprint}
\begin{abstract}
Reliable provenance for LLM outputs requires multi-bit watermarks that remain robust under editing while maintaining low false-positive rates. Existing ECC-based LLM watermarks rely on hard-decision decoding, discarding token-level reliability information and limiting robustness under post-generation edits. We propose \textbf{CORE-BREW}, a \textbf{CO}nstant-hit-\textbf{R}ate \textbf{E}mbedding extension of BREW for multi-bit watermarking. CORE-BREW calibrates the watermark channel by targeting a fixed hit rate $p^\star$, yielding closed-form per-token log-likelihood ratios (LLRs) for soft-decision decoding. It incorporates entropy-aware erasures to limit perturbations in low-entropy contexts and combines likelihood-based scoring with soft-decision list decoding to exploit soft evidence. Experiments on open-source LLMs under token-level edits and paraphrasing demonstrate that CORE-BREW generally improves detection robustness and payload recovery over the BREW baseline while maintaining low observed false-positive rates. Despite higher conditional perplexity, BLEU and BERTScore remain close to those of unwatermarked text, indicating comparable reference-based translation quality.
\end{abstract}


\section{Introduction}

Large language models (LLMs) have been widely adopted for large-scale text generation in education, journalism, software development, and public policy. This rapid deployment raises critical questions regarding provenance, accountability, and safety: when a text appears online, determining its origin, algorithmic provenance, and generation context becomes essential. Watermarking has emerged as a promising mechanism to embed verifiable provenance signals at generation time~\cite{kirchenbauer2023watermark} rather than relying solely on post-hoc classifiers~\cite{mitchell2023detectgpt}.

Recent work has moved beyond zero-bit ``presence tests'' toward \emph{multi-bit} watermarks that embed payloads such as user identifiers or policy tags. A particularly effective line couples keyed token partitioning with error-correcting codes (ECCs) and \emph{designated-codeword} verification, achieving strong false positive rate (FPR) control by accepting only when a pre-specified codeword is recovered in sufficiently many blocks~\cite{qu2025provably,chao2024watermarking,kim2026blockwise}. In these schemes, the detector typically maps each generated token into a single bit (e.g., membership in a keyed green/red list) and applies a bounded-distance ECC decoder. 
This hard-decision interface is simple and provides a transparent acceptance rule, but it discards token-level reliability information, which can limit robustness under synonym substitution, insertion/deletion misalignment, or paraphrasing.

A principled approach to addressing this limitation is \emph{soft-decision} decoding, which substantially improves decoding reliability in classical communication systems when calibrated log-likelihood ratios (LLRs) are available~\cite{lin2004error,richardson2008modern}. However, directly integrating soft decoding into existing LLM watermarking frameworks presents fundamental challenges. Under conventional fixed logit biasing, the induced channel from hidden watermark bits to observed token categories is highly context-dependent: the probability of sampling from the intended token list varies substantially across prompts and generation steps. Consequently, a fixed token observation does not correspond to a consistent reliability level across positions, hindering principled LLR construction. 

We propose \textbf{CORE-BREW}, a \textbf{CO}nstant-hit-\textbf{R}ate \textbf{E}mbedding extension of BREW that addresses this problem by explicitly shaping the watermark channel. On the embedding side, CORE-BREW introduces a \emph{Constant Hit-Rate} rule that adaptively selects a logit offset at each generation step so that the total probability mass assigned to the target token list is fixed to $p^\star \in (0.5,1)$, independent of the local context. In the fully embedded regime, where no distortion safeguard is triggered, this calibration makes each token position behave approximately as a position-homogeneous binary symmetric channel (BSC) with crossover probability $1-p^\star$. Consequently, each observed token category yields a closed-form per-token LLR with constant magnitude $\log\frac{p^\star}{1-p^\star}$, with its sign determined by keyed list membership, providing a principled soft-information interface for ECC-based multi-bit watermarking. Because enforcing a constant hit rate can require large logit shifts in low-entropy contexts, CORE-BREW introduces distortion guards that cap the bias magnitude and skip embedding at low-entropy or highly peaked positions. These positions are treated as \emph{erasures} with zero LLR during detection, limiting excessive perturbations of the model distribution.

On the detection side, CORE-BREW incorporates the calibrated reliability information into the block-wise BCH and window-shifting framework of BREW~\cite{kim2026blockwise}. For each candidate shift, CORE-BREW constructs bit-level LLRs and combines likelihood-based scoring with lightweight soft-decision list decoding to evaluate the block-specific designated codeword. This allows the detector, under moderate perturbation regimes, to recover the designated codeword even when the raw hard-decision representation lies outside the bounded-distance decoding sphere. 
Because this soft-decision rule may expand the acceptance region beyond that of BREW, we analyze the resulting acceptance rule separately and evaluate its FPR--TPR trade-off.
To isolate the contribution of soft-decision decoding from Constant Hit-Rate channel shaping, we evaluate a \emph{hard-decision ablation} that retains the same calibrated embedding and distortion safeguards but replaces likelihood-based scoring and list decoding with bounded-distance designated-codeword verification.

We evaluate CORE-BREW on open-source LLMs under clean generation, token-level substitution, insertion and deletion, and paraphrasing, comparing against multi-bit watermarking baselines~\cite{yoo2024advancing,qu2025provably,kim2026blockwise}. The results show that CORE-BREW improves low-FPR discrimination over the hard-decision BREW baseline in settings where calibrated soft evidence is beneficial, while matching unwatermarked generation in BLEU and BERTScore.

\paragraph{Contributions.}
Our main contributions are as follows.
First, we propose CORE-BREW, which extends BREW with Constant Hit-Rate embedding and entropy-aware erasures to induce a calibrated watermark channel and enable closed-form per-token LLRs.
Second, we combine calibrated LLR scoring with lightweight soft-decision list decoding and score-based thresholding to improve detection while controlling the FPR--TPR trade-off.
Finally, across diverse models, datasets, token-level attacks, and paraphrasing, CORE-BREW consistently improves low-FPR detection performance over BREW; furthermore, ablations isolate the distinct advantage of soft decoding while generation quality (BLEU and BERTScore) remains comparable to unwatermarked text.

\section{Background and Related Work}
\subsection{Watermarking for LLM-generated text}
Text watermarking methods for large language models are broadly divided into (i) \emph{zero-bit} schemes that test only for watermark presence and (ii) \emph{multi-bit} schemes that embed recoverable payloads. A canonical zero-bit approach is the keyed green--red partition method of Kirchenbauer et al.~\cite{kirchenbauer2023watermark}, which biases sampling toward a secret ``green'' subset and detects by statistical tests on green-token counts. Subsequent work studies robustness, distributional fidelity, and detectability limits of such partition-based signals~\cite{wu2024dipmark,zhao2024provable,takezawa2023necessary}. Post-hoc detectors such as DetectGPT provide complementary evidence but do not offer keyed provenance guarantees~\cite{mitchell2023detectgpt}.

\subsection{Multi-bit Watermarking with Error-Correcting Codes}
Early designs, such as MPAC~\cite{yoo2024advancing}, demonstrate the feasibility of multi-bit payload embedding via direct token-level signaling, but their recovery reliability can degrade under editing and paraphrasing attacks. This motivates ECC-based approaches that introduce redundancy for robust payload recovery. Qu et al.~\cite{qu2025provably} study multi-bit watermarking with classical ECCs and a security-oriented evaluation protocol emphasizing robustness to edits, while Chao et al.~\cite{chao2024watermarking} adopt LDPC-style and window-based constructions to improve recovery on short texts.

A challenge in ECC-based multi-bit watermarking is false-positive control. If a detector accepts any decoded codeword, or relies on nearest-codeword recovery without a designated target, watermark-free text may be mapped to a valid payload, inflating the false positive rate (FPR). BREW~\cite{kim2026blockwise} addresses this issue through a block-wise designated-codeword architecture. It partitions the sequence into fixed-length blocks and assigns each block a key-dependent BCH codeword. During generation, the bits of the designated codeword are embedded through keyed token partitions, while detection maps tokens back to hard bits and applies bounded-distance BCH decoding. A block is accepted only when the decoded result matches its designated codeword, and the text is declared watermarked when a sufficient fraction of blocks are matched. BREW employs window shifting to compensate for block-boundary misalignment caused by insertion and deletion attacks. By restricting acceptance to designated codewords rather than arbitrary valid codewords, BREW provides false-positive control. However, its hard-decision decoding discards token-level reliability information, motivating the calibrated soft-information framework developed in this work.

\subsection{Soft information, channel modeling, and decoding}
In classical coding theory, soft-decision decoding improves reliability when calibrated log-likelihood ratios (LLRs) are available. For algebraic block codes such as BCH, Chase-type and ordered-statistics algorithms provide practical realizations of soft decoding~\cite{chase1972,lin2004error,richardson2008modern}. However, most ECC-based LLM watermarking methods rely on hard decisions, mapping each token to a single bit before decoding~\cite{qu2025provably,chao2024watermarking,kim2026blockwise}. Under fixed logit biasing, the induced watermark channel remains context-dependent, making principled LLR construction difficult without explicit calibration.

\subsection{Attacks, paraphrasing, and evaluation protocols}
Robust watermark evaluation considers both token-level edits (e.g., synonym substitution, insertion, deletion) and higher-level paraphrasing or rewriting~\cite{Wieting2018,morris-etal-2020-textattack,krishna2023paraphrasing}. Prior work shows that sufficiently strong paraphrasing can substantially degrade watermark signals, motivating evaluation across a range of attack strengths rather than at a single operating point~\cite{wolff2020attacking}. From a cryptographic perspective, pseudorandom code constructions highlight inherent trade-offs between robustness and indistinguishability~\cite{christ2024pseudorandom}. Together, this literature suggests that practical watermark systems should (i) analyze detection performance in terms of the FPR--TPR trade-off, (ii) characterize robustness under both token-level and semantic-level attacks, and (iii) clearly specify the acceptance criterion used to control false positives. Our evaluation follows these principles by explicitly analyzing the FPR--TPR trade-off.

\section{Problem Setup \& Baseline}
\label{sec:setup}

We formulate watermark detection as a hypothesis test between watermarked ($H_1$) and unwatermarked ($H_0$) text, evaluated by the true positive rate (TPR) and false positive rate (FPR). We use BREW~\cite{kim2026blockwise} as the designated-codeword baseline. BREW partitions generated text into blocks and accepts a block only when its hard-decision representation falls within the bounded-distance decoding region of the corresponding designated BCH codeword. This acceptance rule provides strong false-positive control under the baseline assumptions but discards token-level reliability information. CORE-BREW addresses this limitation by calibrating the induced watermark channel and exploiting soft information during detection. Formal details of BREW are provided in Appendix~\ref{app:baseline-details}.


\section{Proposed Method}
\label{sec:method}  

\begin{figure}[t]
    \centering
    \includegraphics[width=\linewidth]{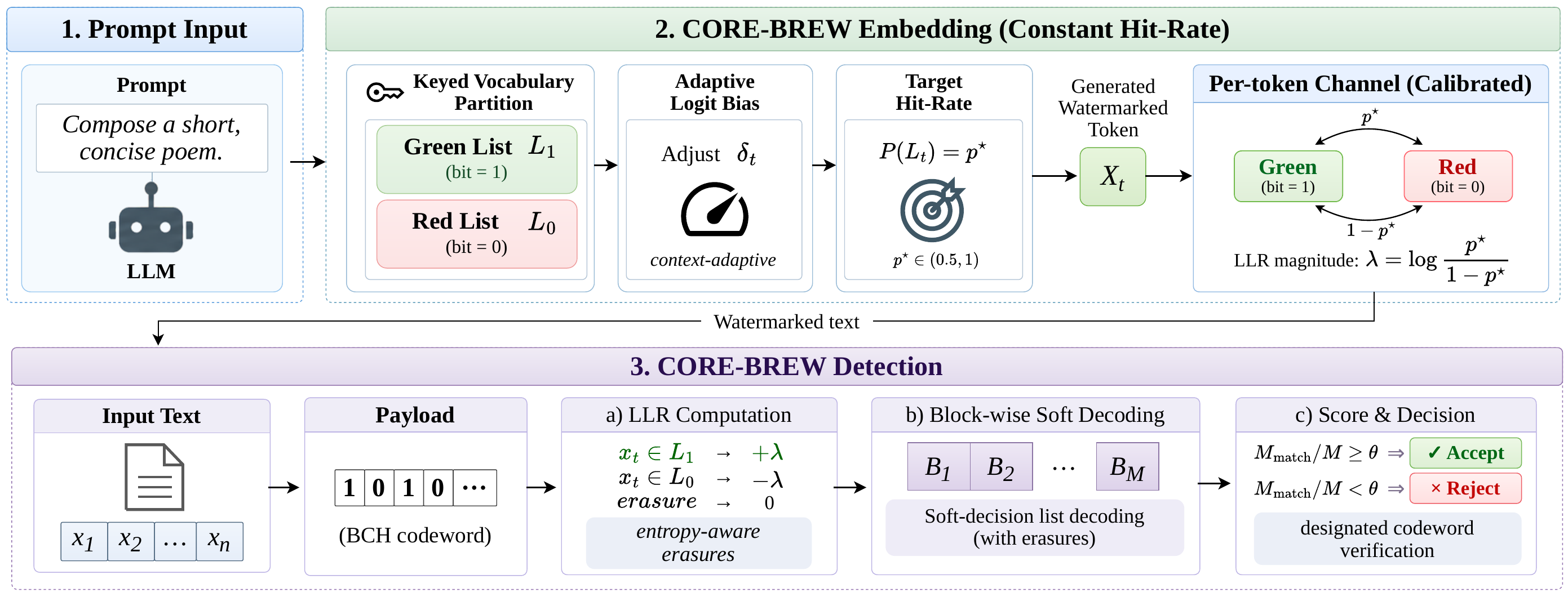}
    \caption{Overview of CORE-BREW. Adaptive logit bias maintains a target hit rate $p^\star$ over a keyed vocabulary partition, providing calibrated token-level soft information. The detector converts token observations into LLRs and performs block-wise soft decoding before the score-based decision.}
    \label{fig:overview}
\end{figure}

We propose \textbf{CORE-BREW}, a \textbf{CO}nstant-hit-\textbf{R}ate \textbf{E}mbedding extension of BREW~\cite{kim2026blockwise} for robust multi-bit LLM watermarking, as illustrated in Figure~\ref{fig:overview}. CORE-BREW retains BREW's block-wise designated-codeword structure while introducing Constant Hit-Rate embedding and erasure-aware LLRs for soft decoding. Likelihood-based scoring and lightweight soft-decision list decoding improve robustness while controlling the FPR--TPR trade-off through a score threshold. The detection procedure is detailed in Section~\ref{sec:method_decode}.

\subsection{Constant Hit-Rate Soft Embedding}
\label{sec:method_embedding}  

At each generation step $t$, let $p_t(v)$ be the base model distribution over the vocabulary $v \in \mathcal{V}$ (after any standard sampling temperature, if used). Given the designated codeword bit $z_t \in \{0,1\}$ at position $t$, we define the \emph{target list} $L_t := L_{z_t}$ from a fixed keyed vocabulary partition $(L_0,L_1)$. Let $m_t := \sum_{v \in L_t} p_t(v)$ be the base model probability mass assigned to the target list.

\paragraph{Constant hit-rate calibration.}
We choose a single global hit-rate target $p^\star \in (1/2,1)$ and compute a
logit offset $\delta_t$ so that the \emph{post-bias} distribution assigns exactly
$p^\star$ total mass to the target list:
\begin{equation}
\delta_t \;=\; \log\frac{p^\star}{1-p^\star} \;-\; \log\frac{m_t}{1-m_t}.
\end{equation}
We then bias logits by adding $\delta_t$ to tokens in $L_t$:

\begin{equation}
\begin{aligned}
\tilde{\ell}_t(v)
&=
\begin{cases}
\ell_t(v)+\delta_t, & v \in L_t,\\
\ell_t(v), & v \notin L_t,
\end{cases}
\qquad
\tilde{p}_t(v)
&=
\frac{\exp(\tilde{\ell}_t(v))}
{\sum_{u\in\mathcal{V}} \exp(\tilde{\ell}_t(u))}.
\end{aligned}
\end{equation}
This makes $\sum_{v\in L_t} \tilde{p}_t(v)=p^\star$ (see Appendix~\ref{app:chr-proof}).
Unlike fixed logit biasing, Constant Hit-Rate fixes the target-list mass to $p^\star$, yielding a position-homogeneous reliability scale for LLR computation.

\paragraph{Entropy-aware erasures for distortion mitigation.}
In practice, when the base distribution exhibits \emph{low entropy} (e.g., a near-deterministic next-token distribution), enforcing a fixed $p^\star$ can require excessively large logit shifts $|\delta_t|$, which degrades fluency and factuality (a known form of degeneration under distributional perturbations~\cite{holtzman2020curious}). To prevent this, we introduce an \textbf{entropy-aware erasure safeguard}:

\begin{itemize}
  \item \textbf{Bias cap:} We first compute the raw constant hit-rate bias and clamp it as $\delta_t \leftarrow \mathrm{clip}(\delta_t, -\delta_{\max}, \delta_{\max})$. We then compute the post-bias target-list mass induced by the clamped bias.
  \item \textbf{Erasure condition (feasibility):} If, after clamping, the induced target-list mass fails to strictly exceed $1/2$ (i.e., the embedding cannot guarantee positive correlation), we \emph{skip embedding} at step $t$ and treat it as an erasure. Equivalently, we set $\delta_t = 0$ and mark this position as erased.
  \item \textbf{Erasure condition (entropy/probability concentration):} We also skip embedding if the base distribution exhibits excessive probability concentration or low entropy, i.e., $\max_v p_t(v)\ge q_{\max}$ or $H_t < H_{\min}$.
\end{itemize}

In decoding, erased positions contribute \emph{zero} LLR (Section~\ref{sec:method_channel}), yielding a binary-symmetric channel with erasures. 
This limits excessive distributional perturbation in low-entropy regions while retaining erasure-aware detection.

\paragraph{Detector Access Assumptions.}
The erasure rule depends on $p_t(\cdot)$, and hence on quantities such as $m_t$, entropy, or peak probability. A detector that uses erasures therefore needs access to the generator or an equivalent scoring model to recompute these quantities from the observed prefix. This model-aware assumption makes CORE-BREW most suitable for provider-side verification, where such access is available. Token-only detection remains possible by disabling erasures, but it does not benefit from the entropy-aware safeguard and is not the focus of our evaluation.

\subsection{Channel Modeling and Per-Token LLRs}
\label{sec:method_channel}  

Define the observed watermark bit at step $t$ as $Y_t := f_K(s_t)\in\{0,1\}$, where $s_t$ is the sampled token and $f_K$ indicates whether $s_t \in L_1$ under the fixed keyed partition.
At non-erased, unclipped positions, Constant Hit-Rate embedding induces an approximately position-homogeneous binary symmetric channel (BSC):
$
\Pr(Y_t = z_t) = p^\star, 
\Pr(Y_t \neq z_t) = 1-p^\star.
$
This justifies a principled per-token log-likelihood ratio (LLR):
\[
\Lambda_t(Y_t) \;=\; \log \frac{\Pr(Y_t \mid z_t=1)}{\Pr(Y_t \mid z_t=0)}
\;=\;
\begin{cases}
\;\;\log\frac{p^\star}{1-p^\star}, & Y_t = 1,\\[4pt]
-\log\frac{p^\star}{1-p^\star}, & Y_t = 0.
\end{cases}
\]
Under attacks, these LLRs represent nominal pre-attack reliability scores rather than exact post-attack likelihoods; online estimation of attack-dependent channel parameters lies outside the scope of this pre-attack calibration model. When an erasure safeguard is triggered at step $t$, we set $\Lambda_t(Y_t)=0$, ensuring that erased positions contribute no evidence.

\subsection{From Token-Level to Bit-Level LLRs}
\label{sec:method_bitllr}  

Let the code length be $n$. For each block $j$ and candidate shift $s$, we construct a bit-level reliability vector $\Lambda^{(j,s)} \in \mathbb{R}^n$ by assigning shifted token positions to BCH bit positions. Let $u_b^{(j,s)}$ denote the token index assigned to bit position $b$ under the shifted candidate block. 
For each bit position, we set
\[
\Lambda^{(j,s)}_b =
\begin{cases}
\Lambda_{u_b^{(j,s)}}\!\left(Y_{u_b^{(j,s)}}\right), 
& \text{if } u_b^{(j,s)} \text{ is valid},\\
0, 
& \text{otherwise}.
\end{cases}
\]
Invalid positions are treated as erasures with zero evidence. The shift $s$ constructs the candidate block representation, while the fixed keyed partition maps each observed token to its hard bit.

\subsection{Window-Shifting Soft-Decision Detection}
\label{sec:method_decode}

Edits and paraphrases can alter token positions and misalign block boundaries. We therefore test candidate shifts $s \in [-s_{\max}, s_{\max}]$ for each block $j$. For each shifted candidate, the detector constructs the corresponding bit-level reliability vector $\Lambda^{(j,s)}$ as described in Section~\ref{sec:method_bitllr}.
For each shift $s$, CORE-BREW computes a block-level score for the designated codeword $c^{(j)}\in\{0,1\}^n$:
\[
S^{(j)}(s)
=
\sum_{b=0}^{n-1}
\left(2c^{(j)}[b]-1\right)
\Lambda^{(j,s)}_b.
\]
The score $S^{(j)}(s)$ measures the likelihood evidence supporting the designated codeword under the calibrated channel.
CORE-BREW then applies a lightweight Chase-style list-decoding procedure ~\cite{chase1972,lin2004error}. Starting from the observed hard-decision vector, the detector identifies the least reliable bit positions according to the magnitudes of their LLRs and enumerates a small set of test error patterns over these positions. Erased positions have zero LLR and are therefore naturally treated as least reliable. Each resulting candidate is passed through the bounded-distance BCH decoder. A block is accepted only if at least one decoded candidate equals the designated codeword and the designated-codeword score satisfies
$
S^{(j)}(s) \ge \tau_{\mathrm{blk}}.
$

This Chase-style search can recover the designated codeword even when the raw hard decision lies outside the BCH bounded-distance radius $t$. The bounded-distance decoder therefore remains a decoding subroutine, while the overall CORE-BREW acceptance rule additionally incorporates calibrated soft evidence through list decoding and score-based rejection.
At the text level, we aggregate block decisions by evaluating the total number of matched blocks, denoted by $M_{\mathrm{match}} \le M$: the text is declared watermarked if $\frac{M_{\mathrm{match}}}{M} \ge \theta$.

\subsection{Design Choices and Practical Considerations}
\label{sec:method_design}  

\paragraph{Selection of Target Hit-Rate $p^\star$.}
The target hit-rate $p^\star$ controls the reliability of the induced bit channel: larger values yield larger LLR magnitudes and stronger watermark evidence, but may require larger logit shifts when $m_t$ is far from the target. We treat $p^\star$ as a tunable parameter, analyzing detection reliability in Section~\ref{sec:experiments} and detailed hit-rate/text-quality trends in Appendices~\ref{app:hit-rate-sweep} and~\ref{app:text-quality-sweep}.

\paragraph{Low-entropy safeguards.}
The entropy-aware erasure mechanism prevents extreme biasing when the model is near-deterministic. By turning difficult positions into erasures (LLR $=0$), the safeguard limits distributional perturbation while allowing the decoder to ignore unreliable evidence. $\delta_{\max}$ and $H_{\min}$ (or $q_{\max}$) control the trade-off between watermark evidence and distributional distortion.


\paragraph{Mechanism for Low-FPR Detection Robustness.}
CORE-BREW improves detection through Chase-style list decoding and erasure-aware LLR scoring; beyond-radius recovery is possible but is not the only source of gain. Its acceptance rule combines bounded shift/candidate search, designated-codeword verification, and score thresholding. The score threshold $\tau_{\mathrm{blk}}$ provides additional filtering rather than serving as a standalone FPR-control mechanism. Section~\ref{sec:Theory} analyzes the acceptance rule, and Section~\ref{sec:experiments} evaluates the FPR--TPR trade-off empirically.

\paragraph{Implementation.}
We compute $m_t$ using masked softmax sums or stable log-sum-exp over the target list and its complement. Decoding scales with $(2s_{\max}+1)$ shifts and a small Chase-style candidate list; detailed thresholds and list-decoding parameters are given in Appendix~\ref{app:impl_repro_details}.


\section{Theoretical Analysis}
\label{sec:Theory}


Our analysis separates two notions conflated in prior multi-bit watermarking work:
(i) bounded-distance acceptance on the detector's hard-decision representation, which admits designated-codeword false-positive analysis under the analysis model, and
(ii) soft-decision detection using calibrated LLR information, which may expand the acceptance region and therefore requires separate analysis.
Under Constant Hit-Rate embedding, non-erased positions approximately follow a position-homogeneous binary symmetric channel with target hit rate $p^\star$, enabling calibrated LLR computation and scoring.
We use BREW's bounded-distance designated-codeword rule as the hard-decision reference, while CORE-BREW combines calibrated block scoring, designated-codeword verification, and limited Chase-style list decoding. Accordingly, false-positive behavior is analyzed separately for the hard- and soft-decision acceptance rules. Detailed hard-decision FPR analysis and score-based bounds for CORE-BREW are provided in Appendices~\ref{app:hard-fpr} and~\ref{app:core-score-bounds}, respectively.

\section{Experiments}
\label{sec:experiments} 

Our experiments address four questions: (Q1) clean detection performance; (Q2) robustness under token-level attacks; (Q3) robustness under paraphrasing and semantic editing; and (Q4) targeted ablations and parameter sensitivity. We use a unified evaluation pipeline based on MarkLLM~\cite{pan-etal-2024-markllm}, with official or publicly available baseline implementations where applicable. Unless otherwise stated, we report mean $\pm$ standard deviation over multiple random seeds.

\subsection{Experimental Setup}
\paragraph{Models and datasets.}
We evaluate OPT-1.3B~\cite{zhang2022} and Mistral-7B~\cite{jiang2023mistral7b} on C4 and OpenGen, where OpenGen consists of 3,000 two-sentence prompts derived from WikiText-103. Unless otherwise specified, OPT-1.3B on C4 is used as the default setting, and we generate continuations of length $T=500$ tokens.

\paragraph{Schemes compared.}
We compare MPAC~\cite{yoo2024advancing}, Qu et al.~\cite{qu2025provably}, BREW~\cite{kim2026blockwise}, and CORE-BREW. We discuss the method of Chao et al.~\cite{chao2024watermarking} as a closely related ECC-based scheme, but do not include it in our empirical comparison because its implementation is not publicly available. We match payload budgets across the evaluated methods and reconfigure BREW to use the same BCH code and block segmentation as CORE-BREW, rather than its original BCH $(31,6,7)$ setting, to control for differences in ECC strength. 

\paragraph{Code, detection parameters, and metrics.}
Unless otherwise stated, we use BCH parameters $(n,k,t)=(63,7,15)$~\cite{lin2004error}, window shifting with $s_{\max}=10$, $p^\star=0.9$, and model-aware detection as the default CORE-BREW configuration. A sequence is declared watermarked if $M_{\mathrm{match}}\ge1$. We report text-level TPR/FPR as mean $\pm$ standard deviation over multiple random seeds. Payload recovery is measured by payload bit accuracy, defined as the fraction of correctly recovered payload bits over all embedded payload bits; unrecovered bits are counted as incorrect. We report PPL, BLEU~\cite{Papineni2002}, and BERTScore~\cite{zhang2020bertscore} for text quality. Additional analyses of code length, window-shift range, payload size, and short-text performance are provided in Appendices~\ref{app:code-length}, \ref{app:Window_shift}, \ref{app:payload_scaling}, and~\ref{app:short_text}. Implementation details for entropy-aware safeguards and list decoding, together with text-quality results, are provided in Appendices~\ref{app:impl_repro_details} and~\ref{app:text-quality-sweep}.

\subsection{Clean Detection Performance (Q1)}
We first evaluate detection performance on clean generations, with particular attention to false positive behavior. Table~\ref{tab:clean} reports TPR and FPR on OPT-1.3B at each method's default or recommended detection threshold. MPAC~\cite{yoo2024advancing} and Qu et al.~\cite{qu2025provably} achieve near-perfect TPR but exhibit substantially higher FPR. BREW~\cite{kim2026blockwise} and CORE-BREW maintain near-zero FPR while retaining high clean-detection TPR. On C4, CORE-BREW performs comparably to BREW, while on OpenGen it improves TPR from $0.9650$ to $0.9817$ with zero observed FPR. These results indicate that CORE-BREW preserves the conservative false-positive behavior of BREW without substantially sacrificing clean detection performance. Corresponding results on Mistral-7B are provided in Appendix~\ref{app:clean-breakdown}.

\begin{table}[t]
\centering
\caption{Clean detection performance on OPT-1.3B. We report mean $\pm$ standard deviation over multiple random seeds at each method's default or recommended detection threshold.}
\label{tab:clean}
\small
\setlength{\tabcolsep}{4pt}
\begin{tabular}{lcccc}
\toprule
& \multicolumn{2}{c}{C4} & \multicolumn{2}{c}{OpenGen} \\
\cmidrule(lr){2-3} \cmidrule(lr){4-5}
Method
& TPR $\uparrow$
& FPR $\downarrow$
& TPR $\uparrow$
& FPR $\downarrow$ \\
\midrule
MPAC
& $1.0000 \pm 0.0000$
& $0.8333 \pm 0.0239$
& $1.0000 \pm 0.0000$
& $0.8917 \pm 0.0062$ \\

Qu et al.~\cite{qu2025provably}
& $1.0000 \pm 0.0000$
& $0.9250 \pm 0.0071$
& $1.0000 \pm 0.0000$
& $0.9533 \pm 0.0143$ \\

BREW
& $0.9383 \pm 0.0143$
& $0.0000 \pm 0.0000$
& $0.9650 \pm 0.0108$
& $0.0000 \pm 0.0000$ \\

CORE-BREW
& $0.9383 \pm 0.0062$
& $0.0000 \pm 0.0000$
& $0.9817 \pm 0.0085$
& $0.0000 \pm 0.0000$ \\
\bottomrule
\end{tabular}
\end{table}


\subsection{Synthetic Token-Level Attacks (Q2)}
\label{sec:synthetic-attacks}

We evaluate robustness against controlled token-level perturbations applied to both watermarked and unwatermarked texts at an attack rate of $\alpha=0.1$ on OPT-1.3B.
We consider three attack categories:
\emph{token-preserving substitution}, which replaces tokens with synonyms via TextAttack-style transformations~\cite{morris-etal-2020-textattack};
\emph{token deletion}, which independently removes tokens with probability $\alpha$; and
\emph{token insertion}, which adds tokens or expands them into multi-token variants.
For insertion and deletion, we apply window-shifting detection with maximum offset $s_{\max}$ to compensate for block-boundary misalignment, whereas substitution preserves token positions and requires no alignment correction.

\begin{figure*}[t]
    \centering
    \includegraphics[width=\textwidth]{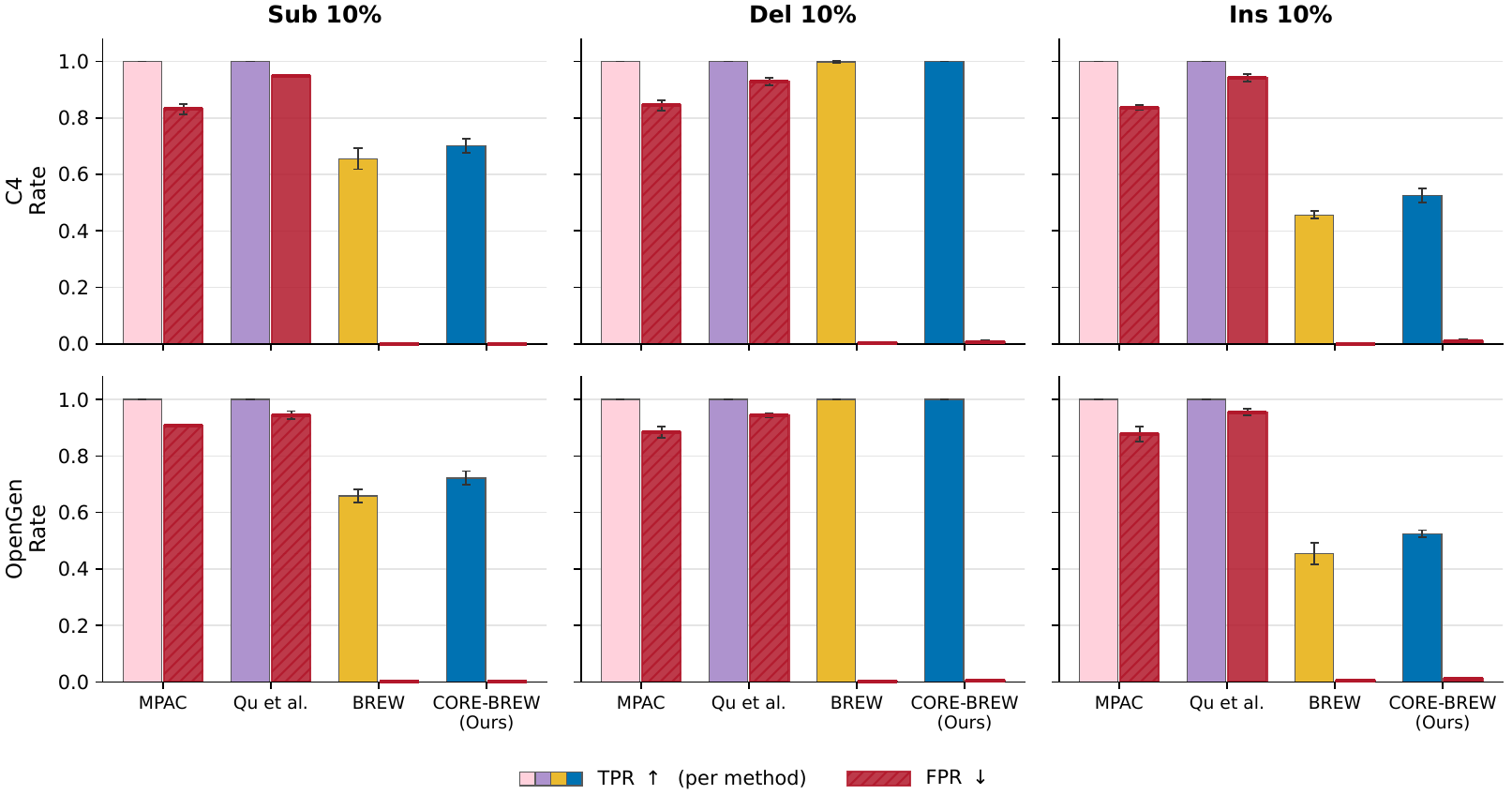}
    \caption{Detection performance under 10\% substitution, deletion, and insertion attacks on OPT-1.3B using C4 and OpenGen. Colored bars denote TPR and red bars denote FPR.}
    \label{fig:synthetic_token_attacks}
\end{figure*}

\paragraph{Results.}
Figure~\ref{fig:synthetic_token_attacks} shows that MPAC~\cite{yoo2024advancing} and Qu et al.~\cite{qu2025provably} retain near-perfect TPR but exhibit consistently high FPR, whereas BREW~\cite{kim2026blockwise} and CORE-BREW maintain low observed FPR. Under substitution, CORE-BREW improves TPR over BREW on both C4 ($0.6550\rightarrow0.7017$) and OpenGen ($0.6583\rightarrow0.7217$), while both methods maintain zero observed FPR. Among the evaluated perturbations, token deletion causes the least performance degradation, with BREW and CORE-BREW achieving near-perfect TPR while maintaining FPR below $1\%$. Insertion remains the most challenging attack, but CORE-BREW again improves TPR over BREW on C4 ($0.4567\rightarrow0.5250$) and OpenGen ($0.4550\rightarrow0.5250$), while maintaining low FPR. 
Detailed results, including payload bit accuracy, are provided in Appendix~\ref{Appendix:Synthetic_Token_Level_Attacks}, with Mistral-7B results in Appendix~\ref{app:mistral-results}.
\FloatBarrier

\subsection{Paraphrasing and Semantic Editing (Q3)}
\label{sec:paraphrasing}

\begin{figure*}[t]
    \centering
    \includegraphics[width=0.88\textwidth]{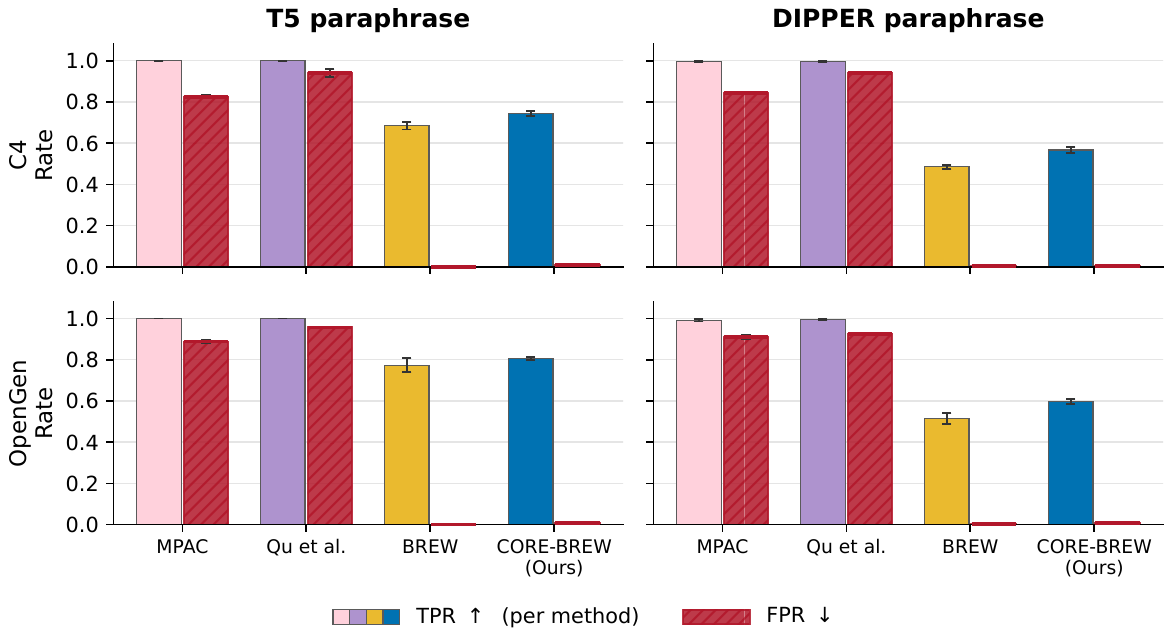}
    \caption{Detection performance under T5 and DIPPER paraphrasing attacks on OPT-1.3B at $T=500$ using C4 and OpenGen. Colored bars denote TPR and red bars denote FPR.}
    \label{fig:Paraphrasing_Attack}
\end{figure*}
We evaluate robustness against paraphrasing attacks that modify surface form while preserving semantic content. We consider both a T5-based paraphraser~\cite{raffel2020exploring} and DIPPER~\cite{krishna2023paraphrasing}, providing two distinct semantic rewriting settings. Paraphrases are filtered to retain semantically related outputs with non-trivial surface variation. We evaluate continuations of length $T=500$ on C4 and OpenGen using $\delta=3$ for the baseline methods and $p^\star=0.9$ for CORE-BREW. Detailed numerical results, including payload bit accuracy, are provided in Appendix~\ref{Appendix:Paraphrasing_attacks}.

Figure~\ref{fig:Paraphrasing_Attack} shows that MPAC~\cite{yoo2024advancing} and Qu et al.~\cite{qu2025provably} retain near-perfect TPR but exhibit consistently high FPR across both datasets and attacks, whereas BREW~\cite{kim2026blockwise} and CORE-BREW maintain low observed FPR. Compared with BREW, CORE-BREW consistently improves TPR under both paraphrasing attacks: from $0.6850$ to $0.7450$ on C4 and from $0.7733$ to $0.8067$ on OpenGen under T5, and from $0.4850$ to $0.5667$ and from $0.5150$ to $0.5983$ under DIPPER, respectively. 
Across all settings, CORE-BREW maintains observed FPR below $1\%$, demonstrating improved robustness to both T5- and DIPPER-based paraphrasing while preserving low FPR.

\subsection{Targeted Ablations and Parameter Sensitivity (Q4)}

\begin{table}[t]
\centering
\caption{Ablation of soft-decision decoding on OPT-1.3B with C4 at $p^\star=0.9$.}
\label{tab:soft_decoding_ablation}
\small
\setlength{\tabcolsep}{4pt}
\renewcommand{\arraystretch}{1.05}
\resizebox{\columnwidth}{!}{%
\begin{tabular}{l l c c c}
\toprule
\textbf{Attack}
& \textbf{Method}
& \textbf{TPR}
& \textbf{FPR}
& \textbf{Payload Bit Acc.} \\
\midrule
\multirow{2}{*}{Clean}
& CORE-BREW w/o Soft Decoding & $0.9067 \pm 0.0062$ & $0.0000 \pm 0.0000$ & $0.8983 \pm 0.0103$ \\
& CORE-BREW & $0.9383 \pm 0.0062$ & $0.0000 \pm 0.0000$ & $0.9308 \pm 0.0042$ \\
\midrule
\multirow{2}{*}{10\% Substitution}
& CORE-BREW w/o Soft Decoding & $0.6617 \pm 0.0417$ & $0.0000 \pm 0.0000$ & $0.5367 \pm 0.0283$ \\
& CORE-BREW & $0.7017 \pm 0.0249$ & $0.0000 \pm 0.0000$ & $0.5492 \pm 0.0120$ \\
\bottomrule
\end{tabular}%
}
\end{table}

We analyze the contribution of soft-decision processing together with the effects of the target hit rate $p^\star$, window-shift budget $s_{\max}$, and entropy-aware erasures. Table~\ref{tab:soft_decoding_ablation} compares CORE-BREW with a matched hard-decision variant that uses the same Constant Hit-Rate embedding and erasure decisions, but replaces LLR-based scoring and list decoding with bounded-distance designated-codeword decoding, filling erased positions via a pseudorandom function (PRF). CORE-BREW improves TPR from $0.9067$ to $0.9383$ under clean generation and from $0.6617$ to $0.7017$ under 10\% substitution, while also improving payload bit accuracy and maintaining zero observed FPR in both settings. These results show that the soft-decision configuration provides improved detection and payload recovery over the corresponding hard-decision reference under the evaluated settings.

Additional parameter studies show that detection and payload recovery improve as $p^\star$ increases and largely saturate at high hit rates, supporting our default choice of $p^\star=0.9$ (Appendix~\ref{app:hit-rate-sweep}). Increasing $s_{\max}$ improves alignment robustness, particularly under insertion, but gradually raises false positives, supporting $s_{\max}=10$ as a practical trade-off (Appendix~\ref{app:Window_shift}). Entropy-aware erasures substantially reduce conditional perplexity with only modest changes in TPR, supporting their role in distortion control (Appendix~\ref{app:erasure-ablation}).

\subsection{Limitations and Threat Model}
\label{sec:limitations}
We focus on provider-side provenance verification, where the detector retains the secret key. A token-only setting uses the observed text, key, and designated codeword schedule without model probabilities, whereas model-aware detection additionally accesses the generator or an equivalent scoring model to reapply the entropy-aware erasure rule to the observed prefix. Under text perturbations, this prefix may differ from that used during generation, so detection-time erasure decisions need not match those made during embedding. We consider token-level edits and moderate paraphrasing, but do not claim robustness to arbitrary semantic rewriting or white-box, key-aware attacks. Model-aware detection incurs additional cost from recomputing token probabilities and erasure decisions, which could be reduced using lightweight proxy models. CORE-BREW also exhibits higher conditional perplexity than unwatermarked generation and the evaluated baselines, although BLEU and BERTScore remain close to those of unwatermarked generation (Appendix~\ref{app:text-quality-sweep}). Improving this likelihood--robustness trade-off, for example through adaptive control of $p^\star$, remains an important direction for future work.

\section{Conclusion}
\label{sec:Conclusion}


We present \textbf{CORE-BREW}, a block-wise multi-bit text watermarking framework combining Constant Hit-Rate embedding, entropy-aware erasures, calibrated LLR scoring, soft-decision decoding, and window shifting for robust detection under text modifications. Across token-level attacks and paraphrasing, CORE-BREW generally improves detection and payload recovery over BREW while maintaining low false-positive rates, with ablations showing gains over the hard-decision variant. The framework also supports payload scaling and shorter generations via adaptive BCH configurations. Although CORE-BREW incurs higher model-likelihood distortion, BLEU and BERTScore remain close to those of unwatermarked generation. Overall, calibrated soft information improves the robustness and flexibility of multi-bit LLM watermarking while preserving practical false-positive control.

\section{Acknowledgments}
This work was partly supported by Institute of Information \& communications Technology Planning \& Evaluation (IITP) grant funded by the Korea government (MSIT) (RS-2026-00399401, Development of Quantum-Safe Infrastructure Migration and Quantum Security Verification Technologies, 50\%) and Institute of Information \& communications Technology Planning \& Evaluation (IITP) grant funded by the Korea government (MSIT) (RS-2026-00442085, Development of V2X Infra Security Core Technologies for Autonomous Vehicle Services, 50\%).

\section{AI Use Statement}

Generative AI tools were used solely to assist with language editing, including grammar checking, translation, and refinement of wording for clarity and readability. All technical content, methodology, experiments, results, and conclusions were developed and verified by the authors. The authors take full responsibility for the final content of this work.

\section{Reproducibility Statement}

We provide the methodological and experimental details required to reproduce CORE-BREW throughout the main paper and appendices. The complete embedding and detection procedures, including Constant Hit-Rate embedding, entropy-aware erasures, LLR construction, window shifting, and soft-decision list decoding, are described in Section~\ref{sec:method} and Appendix~\ref{app:algorithms}. Experimental models, datasets, BCH parameters, attack settings, evaluation metrics, and default hyperparameters are specified in Section~\ref{sec:experiments}, with additional implementation details provided in Appendix~\ref{app:impl_repro_details}. Parameter sensitivity, component ablations, payload scaling, short-text evaluation, and additional model results are reported in the supplementary experiments. Theoretical assumptions and proofs are provided in Appendix~\ref{app:theory-proofs}. We report results over multiple random seeds as mean $\pm$ standard deviation to characterize experimental variability. An anonymized implementation and configuration files are provided in the supplementary material.

\bibliography{iclr2027_conference}
\bibliographystyle{iclr2027_conference}

\appendix

\section{Baseline Full Details}
\label{app:baseline-details}

\subsection{Block-wise Designated-Codeword Watermarking}

This appendix summarizes the block-wise designated-codeword watermarking framework underlying BREW~\cite{kim2026blockwise}, along with relevant ECC-based components used in prior multi-bit watermarking schemes~\cite{qu2025provably,chao2024watermarking}. The framework combines keyed token partitions with error-correcting codes (ECCs) and accepts a block only when the decoded codeword matches its designated codeword. While the framework can be instantiated with different block codes, in this work we use a binary BCH code with parameters $(n,k,t)=(63,7,15)$, as discussed in Appendix~\ref{app:code-length}.

\paragraph{Scope of this appendix.}
This appendix describes the generic block-wise designated-codeword framework used as the baseline abstraction and for the hard-decision false-positive analysis. The analysis adopts an idealized detector-facing bit model to make the acceptance-region argument explicit. Our evaluated CORE-BREW instantiation in Appendix~\ref{app:algorithms}, however, uses a fixed keyed vocabulary partition shared across positions, keeping the token-to-bit mapping invariant across candidate shifts under insertion/deletion attacks. Accordingly, the analysis below should be interpreted as an idealized designated-codeword model rather than as an exact independence guarantee for the evaluated fixed-partition implementation.

\paragraph{Code and blocks.}
Let $C \subseteq \{0,1\}^n$ be a binary BCH code of length $n$, dimension $k$, and minimum Hamming distance $d_{\min}$. Its unique-decoding radius under bounded-distance decoding is
\[
t = \left\lfloor \frac{d_{\min}-1}{2} \right\rfloor.
\]
Generation positions after the prompt are partitioned into blocks of length $n$. Let $\tau \ge N_p$ denote an absolute generation position and define its continuation-relative position as
\[
r = \tau - N_p.
\]
The corresponding block index and within-block offset are
\[
j = \left\lfloor \frac{r}{n} \right\rfloor,
\qquad
b = r \bmod n,
\]
where $j \in \{0,1,\ldots\}$ is the block index and $b \in \{0,\ldots,n-1\}$ is the within-block offset.

\paragraph{Designated codewords.}
For each block $j$, the embedder and detector share a designated codeword $c^{(j)} \in C$, derived from the payload and the secret key. A typical construction applies a per-block pseudorandom mask before encoding to balance codeword weights and avoid fixed patterns~\cite{qu2025provably,kim2026blockwise}. During detection, the decoded block is accepted only when it matches the designated codeword $c^{(j)}$ for that block; arbitrary valid decoded codewords are not accepted.

\paragraph{Keyed token partitions.}
A generic designated-codeword construction may use block-dependent keyed vocabulary partitions. For each block $j$, the vocabulary $\mathcal{V}$ is partitioned into two disjoint lists $L^{(j)}_0$ and $L^{(j)}_1$ (``red/green'') using a cryptographic hash seeded by $(K,j)$:
\begin{align*}
L^{(j)}_0
&=
\{v \in \mathcal{V} :
H(K \Vert j \Vert v)\bmod 2 = 0\},\\
L^{(j)}_1
&=
\{v \in \mathcal{V} :
H(K \Vert j \Vert v)\bmod 2 = 1\}.
\end{align*}
Here, $H(\cdot)$ is modeled as a pseudorandom mapping for analysis~\cite{Goldreich2001,kirchenbauer2023watermark,qu2025provably}. Membership in these lists induces a deterministic mapping
\[
f_j : \mathcal{V} \rightarrow \{0,1\},
\]
which is used during both embedding and detection. Appendix~\ref{app:algorithms} specifies the fixed-partition instantiation used in our evaluated CORE-BREW implementation.

\subsection{Baseline Embedding with Fixed Logit Bias}

At each generation position $\tau$ in block $j$ with within-block offset $b$, the designated code bit to embed is
\[
z_\tau = c^{(j)}[b] \in \{0,1\}.
\]
Let $\ell_\tau(v)$ denote the base-model logit for token $v \in \mathcal{V}$. Under the block-dependent baseline notation used in this appendix, the target list is $L^{(j)}_{z_\tau}$. The baseline applies a fixed logit bias $\delta>0$ to tokens in this target list:
\[
\tilde{\ell}_\tau(v)
=
\begin{cases}
\ell_\tau(v)+\delta, & v\in L^{(j)}_{z_\tau},\\
\ell_\tau(v), & \text{otherwise}.
\end{cases}
\]

This follows the partition-based biasing paradigm~\cite{kirchenbauer2023watermark} adapted to BREW's block-wise designated-codeword setting~\cite{kim2026blockwise}. Because the same bias $\delta$ is applied regardless of the base probability mass assigned to the target list, the resulting target-hit probability varies with context. Consequently, the induced bit channel is generally context-dependent and uncalibrated, limiting the use of principled LLR-based decoding. This limitation motivates the Constant Hit-Rate embedding introduced in CORE-BREW.

\subsection{Baseline Detection with Window Shifting and Bounded-Distance Decoding}

Given a token sequence containing a prompt of length $N_p$ followed by the generated continuation, the detector maps continuation tokens to hard bits using the keyed vocabulary partition. For block $j$ and within-block offset $b$, the corresponding absolute token position is
\[
\tau = N_p + jn + b,
\]
and the detector sets
\[
B^{(j)}[b] = f_j(s_\tau) \in \{0,1\},
\]
yielding a hard-decision block vector $B^{(j)} \in \{0,1\}^n$.

\paragraph{Window shifting for insertions/deletions.}
Insertions and deletions can shift the observed token positions relative to the expected block boundaries. To compensate for such local misalignment, the detector searches over a bounded set of offsets around the nominal block anchor $N_p + jn$. Let $s_{\max}\in\mathbb{N}$ and define
\[
\mathcal{S}
=
\{-s_{\max},\ldots,-1,0,1,\ldots,s_{\max}\},
\qquad
S = |\mathcal{S}| = 2s_{\max}+1.
\]
For each $s\in\mathcal{S}$, the detector constructs a shifted candidate block
$B^{(j,s)}\in\{0,1\}^n$ as
\[
B^{(j,s)}[b]
=
f_j\!\left(s_{N_p + jn + s + b}\right),
\qquad
b=0,\ldots,n-1.
\]

\paragraph{Bounded-distance BCH decoding and designated-codeword verification.}
For each shift $s\in\mathcal{S}$, the detector applies a bounded-distance BCH decoder to the candidate block vector $B^{(j,s)}$. Let
\[
\hat{c}^{(j)}_s
=
\operatorname{Dec}_{\mathrm{BCH}}\!\left(B^{(j,s)}\right)
\]
denote the decoder output. Block $j$ is counted as a match if at least one tested shift decodes exactly to its designated codeword:
\[
\exists s\in\mathcal{S}
\quad \text{s.t.} \quad
\hat{c}^{(j)}_s = c^{(j)}.
\]
Thus, the detector does not accept arbitrary valid BCH codewords; acceptance requires equality with the designated codeword assigned to the block.

Over $M$ evaluated blocks, let $M_{\mathrm{match}}$ denote the number of matched blocks and define the match ratio
\[
\rho
=
\frac{M_{\mathrm{match}}}{M}.
\]
The text is declared watermarked if
\[
\rho \ge \theta,
\]
where $\theta\in(0,1]$ is the text-level detection threshold.

\subsection{Baseline False Positive Behavior and Analysis}
The designated-codeword architecture provides a transparent FPR control mechanism under the standard analysis model. 
Under $H_0$ and pseudorandom partition assumptions, we model the detector-facing candidate block vector as approximately uniform over $\{0,1\}^n$~\cite{qu2025provably,kim2026blockwise}. 
For a fixed alignment, the probability that such a random block vector lies within Hamming distance at most $t$ of a fixed designated codeword is
\[
p_0 = \frac{V_2(n,t)}{2^n},\qquad
V_2(n,t) = \sum_{i=0}^{t}\binom{n}{i},
\]
the volume of a Hamming ball in $\{0,1\}^n$~\cite{lin2004error}. 
With window shifting, a union bound gives
\[
p^{(\mathrm{shift})}_0 \le S\,p_0.
\]
For aggregate text-level bounds, we use an idealized approximate block-independence model. 
Under this approximation, the number of spurious matches under $H_0$ can be modeled as a binomial random variable with success probability $p^{(\mathrm{shift})}_0$. 
This gives the standard Chernoff-style intuition that, when $\theta > p^{(\mathrm{shift})}_0$, the aggregate FPR decays exponentially with the number of blocks $M$~\cite{lin2004error}. 
For the fixed-partition instantiation in Appendix~\ref{app:algorithms}, this should be interpreted as an idealized analysis tool complemented by empirical FPR evaluation on unwatermarked text, since repeated tokens and natural-language correlations can induce dependencies across positions and blocks.

\paragraph{Baseline limitations (motivating calibrated soft information).}
The baseline's guarantees are tied to its \emph{acceptance region}: it accepts a block only when the detector-facing hard-decision vector falls within the bounded-distance decoding neighborhood of the designated codeword. 
This yields combinatorial-style FPR control under the analysis model, but it also (i) discards token reliability information and (ii) operates over a context-dependent, uncalibrated induced channel due to fixed bias $\delta$. 
These limitations motivate the channel-calibrated embedding and LLR-based soft decoding introduced in CORE-BREW.

\section{Algorithmic Details}
\label{app:algorithms} 

This appendix provides concrete algorithmic descriptions of the main components of CORE-BREW. It follows the block-wise designated-codeword framework underlying BREW~\cite{kim2026blockwise}, together with relevant ECC-based components from prior multi-bit watermarking methods~\cite{qu2025provably,chao2024watermarking}. CORE-BREW combines Constant Hit-Rate embedding with erasure-aware LLR construction, likelihood-based scoring, and erasure-aware Chase-style list decoding. We additionally describe a hard-decision ablation that uses the same calibrated embedding and erasure decisions but replaces LLR-based scoring and list decoding with bounded-distance designated-codeword verification on a PRF-filled hard-decision representation. Both variants use a fixed keyed vocabulary partition and block-specific designated BCH codewords determined by the embedded payload and secret key.

\subsection{Keyed Partitions and Codeword Assignment}
\label{app:blockwise}

We use the notation and basic setup introduced in Sections~\ref{sec:setup}--\ref{sec:method}. Let $\mathcal{V}$ denote the vocabulary, $K$ the secret key shared by the embedder and detector, and $C\subseteq\{0,1\}^n$ a binary BCH code of length $n$ and dimension $k$. We use a cryptographic hash function $H$, modeled as a random oracle for analysis~\cite{Goldreich2001}.

\paragraph{Fixed keyed partitions.}
In the implementation used for our experiments, we use a single keyed vocabulary split that is shared across all block positions. 
Given a secret key $K$, we define
\[
L_0 = \{v\in\mathcal{V}: H(K\|v)\bmod 2 = 0\},\qquad
L_1 = \{v\in\mathcal{V}: H(K\|v)\bmod 2 = 1\}.
\]
This follows the keyed green--red paradigm~\cite{kirchenbauer2023watermark}, but unlike position-dependent partition schedules, the token-to-bit mapping is independent of the bit offset. 

This fixed mapping keeps the token-to-bit assignment invariant across candidate shifts, because the same token maps to the same hard bit even when its position changes.

\paragraph{Designated codeword assignment.}
For each block $j$, the assigned payload chunk $m^{(j)}$ is encoded as $c^{(j)}=\mathrm{Encode}_C(m^{(j)})$. Under the default schedule, payload chunks are deterministically selected from the secret key and block index; an explicit payload schedule may instead be supplied. The resulting block-specific codeword serves as the designated target during both embedding and detection. During detection, each reconstructed block is verified specifically against its assigned codeword $c^{(j)}$; arbitrary valid BCH codewords are not accepted. This block-dependent assignment supports multi-bit payload embedding while retaining designated-codeword verification.

\subsection{Constant Hit-Rate Soft Embedding with Distortion-Aware Erasures}
Algorithm~\ref{alg:chr_embedding} summarizes online embedding for a single sequence. At each generation step, it determines the current block and code-bit offset, computes the base target-list mass $m_t$, and chooses a Constant Hit-Rate bias subject to distortion guards. If the clamped bias cannot raise the target-list mass above $1/2$, if the base distribution is too peaked, or if the base entropy is below a minimum threshold, watermark biasing is skipped and the position is marked as an erasure. This limits excessive distributional perturbation and supports the calibrated BSC approximation used in Section~\ref{sec:Theory}. 
When model-aware erasures are enabled, the detector reapplies the same erasure rule through $\mathrm{EraseCheck}$ in Algorithm~\ref{alg:detect} using the observed prefix, assigning zero LLR to positions that satisfy the distortion safeguards at detection time.

\begin{algorithm}[tb]
\caption{Constant Hit-Rate block-wise embedding with distortion-aware erasures}
\label{alg:chr_embedding}
\begin{algorithmic}
\STATE {\bfseries Input:}
prompt tokens $(s_0,\ldots,s_{N_p-1})$, secret key $K$, BCH code $C$, 
block-level payload chunks $\{m^{(j)}\}$, block length $n$, 
target hit-rate $p^\star$, language model $\mathrm{LM}$, 
bias cap $\delta_{\max}$, peak threshold $q_{\max}$, 
entropy threshold $H_{\min}$, numerical clamp $\varepsilon$

\STATE {\bfseries Output:}
watermarked continuation tokens $(s_{N_p},\ldots,s_{N_p+T-1})$

\STATE Construct the keyed fixed partition $L_0,L_1$ from $K$

\FOR{$t = 0$ {\bfseries to} $N_p-1$}
    \STATE Emit prompt token $s_t$ without modification
\ENDFOR

\FOR{$t = N_p$ {\bfseries to} $N_p+T-1$}
    \STATE $(j,b) \gets \mathrm{divmod}(t-N_p,n)$

    \STATE $c^{(j)} \gets \mathrm{Encode}_{C}(m^{(j)})$
    \STATE $z_t \gets c^{(j)}[b]$
    \STATE $L_t \gets L_{z_t}$

    \STATE Query $\mathrm{LM}$ to obtain logits $\ell_t(v)$
    \STATE $p_t(v) \gets \mathrm{softmax}(\ell_t)(v)$
    \STATE $p_{\max} \gets \max_v p_t(v)$
    \STATE $H_t \gets -\sum_{v\in\mathcal{V}} p_t(v)\log p_t(v)$
    \STATE $m_t \gets \sum_{v\in L_t} p_t(v)$
    \STATE $m_t \gets \min(\max(m_t,\varepsilon),1-\varepsilon)$

    \STATE $\delta_t^{\mathrm{raw}} \gets
    \log\!\frac{p^\star}{1-p^\star}
    -
    \log\!\frac{m_t}{1-m_t}$

    \STATE $\delta_t \gets
    \mathrm{clip}(\delta_t^{\mathrm{raw}},-\delta_{\max},\delta_{\max})$

    \STATE $\widetilde{m}_t \gets
    \dfrac{e^{\delta_t}m_t}
    {e^{\delta_t}m_t + (1-m_t)}$

    \IF{$p_{\max} \ge q_{\max}$ {\bfseries or}
        $H_t < H_{\min}$ {\bfseries or}
        $\widetilde{m}_t \le 1/2$}
        \STATE $\delta_t \gets 0$
        \STATE Mark position $t$ as erased
    \ENDIF

    \STATE Bias logits toward $L_t$ using $\delta_t$
    \STATE Sample $s_t$ from the biased distribution
\ENDFOR

\end{algorithmic}
\end{algorithm}

\paragraph{Relation to BREW's fixed-bias embedding.}
BREW uses a fixed logit bias $\delta$ rather than the adaptive Constant Hit-Rate bias $\delta_t$. As a result, the target-list probability varies with context rather than being calibrated to a fixed target hit rate $p^\star$.

\FloatBarrier

\subsection{Bit Sequence Extraction and Block Reconstruction}

During detection, observed tokens are first mapped back to bit sequences using the same fixed keyed partition derived from $K$. Algorithm~\ref{alg:extract_bits} summarizes this mapping for the unshifted case. For each complete block, the detector reuses the fixed partition $(L_0,L_1)$ and maps each token to a hard bit according to its membership in $L_1$. The resulting vector $B^{(j)}\in\{0,1\}^n$ provides the token-induced hard-decision representation. At non-erased positions, these bits agree with the signs of the CORE-BREW LLRs, while erased positions are handled separately by the corresponding decoder. In the hard-decision ablation, erased positions are filled with key-derived PRF bits before bounded-distance decoding. In CORE-BREW, the same token-to-bit observations provide the initial hard decisions for soft decoding, while the associated LLRs encode their reliability. Shifted candidates in Algorithm~\ref{alg:detect} use the same token-to-bit mapping and are evaluated within the allowed shift budget.

\begin{algorithm}[tb]
  \caption{Bit sequence extraction and block reconstruction}
  \label{alg:extract_bits}
  \begin{algorithmic}
    \STATE {\bfseries Input:}
    text tokens $(s_0,\ldots,s_{N-1})$ prompt length $N_p$, secret key $K$, block length $n$
    \STATE {\bfseries Output:}
    complete block-wise bit sequences $B^{(0)},\ldots,B^{(M-1)} \in \{0,1\}^n$

    \STATE $U \gets N - N_p$\COMMENT{number of post-prompt positions}
    \STATE $M \gets \left\lfloor U/n \right\rfloor$ \COMMENT{number of complete blocks}
    \IF{$M \le 0$}
        \STATE \textbf{return} $[\,]$
    \ENDIF

    \STATE Construct the fixed keyed partition $(L_0,L_1)$ from $K$
    \FOR{$j = 0$ {\bfseries to} $M-1$}
        \STATE Initialize $B^{(j)}$ as an array of length $n$
    
        \FOR{$b = 0$ {\bfseries to} $n-1$}
            \STATE $t \gets N_p + jn + b$
            \STATE $B^{(j)}[b] \gets \mathbb{I}[s_t \in L_1]$
        \ENDFOR
    \ENDFOR

    \STATE \textbf{return} $(B^{(0)},\ldots,B^{(M-1)})$
  \end{algorithmic}
\end{algorithm}

\FloatBarrier

\subsection{LLR Computation and Window-Shifting Detection}
\label{app:window-detection}

Algorithm~\ref{alg:detect} summarizes the window-shifting detection procedure for CORE-BREW and its hard-decision ablation. 
On non-erased positions, the calibrated hit-rate model assigns the constant LLR magnitude
\[
\lambda=\log\!\left(\frac{p^\star}{1-p^\star}\right).
\]
CORE-BREW uses these LLRs through score-based acceptance and limited Chase-style list decoding, with erased positions assigned zero reliability. To evaluate the contribution of soft-decision processing, the hard-decision ablation uses the same Constant Hit-Rate embedding and erasure decisions, but fills erased positions with key-derived PRF bits and applies bounded-distance designated-codeword verification instead of LLR-based scoring and list decoding.

\begin{algorithm}[tb]
\caption{Window-Shifting Detection with Block-Dependent Designated Codewords}
\label{alg:detect}
\begin{algorithmic}[1]

\STATE \textbf{Input:}
tokens $(s_0,\ldots,s_{N-1})$, prompt length $N_p$, key $K$, BCH code $C$,
designated codewords $\{c^{(j)}\}$, block length $n$, max shift $s_{\max}$,
text threshold $\theta$, hit rate $p^\star$,
mode $\in\{\texttt{hard},\texttt{soft}\}$

\STATE \textbf{Parameters:}
decoding radius $t_{\mathrm{BCH}}$, block threshold $\tau_{\mathrm{blk}}$,
list-decoding parameters $L_{\max},\ell$

\STATE Construct the fixed keyed partition $(L_0,L_1)$ from $K$
\STATE $\lambda \gets \log\!\left(\frac{p^\star}{1-p^\star}\right)$
\STATE $U \gets N-N_p$
\STATE $M \gets \left\lfloor U/n \right\rfloor$

\IF{$M=0$}
    \STATE \textbf{output} \texttt{not watermarked}
    \STATE \textbf{return}
\ENDIF

\STATE $M_{\mathrm{match}} \gets 0$

\FOR{$j=0$ {\bfseries to} $M-1$}
    \STATE $\mathrm{matched} \gets \texttt{false}$

    \FOR{each shift $s\in\{-s_{\max},\ldots,s_{\max}\}$}
        \STATE $(B^{(j,s)},E^{(j,s)},\Lambda^{(j,s)}) \gets
        \mathrm{ConstructObservation}(s,j,c^{(j)},K,\lambda)$

        \IF{$\mathrm{mode}=\texttt{hard}$}
            \STATE Form $\widetilde{B}^{(j,s)}$ from $B^{(j,s)}$,
            replacing positions marked in $E^{(j,s)}$ with
            key-derived PRF bits

            \IF{$d_H(\widetilde{B}^{(j,s)},c^{(j)})
                \le t_{\mathrm{BCH}}$}
                \STATE $\mathrm{matched}\gets\texttt{true}$
                \STATE \textbf{break}
            \ENDIF

        \ELSE
            \STATE $S^{(j)}(s)
            \gets
            \sum_{b=0}^{n-1}
            (2c^{(j)}[b]-1)\Lambda_b^{(j,s)}$

            \IF{$S^{(j)}(s)\ge\tau_{\mathrm{blk}}$}
                \STATE $\widehat{\mathcal{C}}\gets
                \mathrm{ChaseListDecode}
                (B^{(j,s)},\Lambda^{(j,s)},C,L_{\max},\ell)$

                \IF{$c^{(j)}\in\widehat{\mathcal{C}}$}
                    \STATE $\mathrm{matched}\gets\texttt{true}$
                    \STATE \textbf{break}
                \ENDIF
            \ENDIF
        \ENDIF
    \ENDFOR

    \IF{$\mathrm{matched}$}
        \STATE $M_{\mathrm{match}}
        \gets M_{\mathrm{match}}+1$
    \ENDIF
\ENDFOR

\STATE \textbf{output} \texttt{watermarked} iff
$\frac{M_{\mathrm{match}}}{M}\ge\theta$

\end{algorithmic}
\end{algorithm}

Here, $\mathrm{ConstructObservation}$ maps observed tokens to hard bits and reapplies the distortion-aware erasure rule through $\mathrm{EraseCheck}$ on the observed prefix. Erased positions retain their token-derived hard bits but are assigned zero LLR.

\subsection{CORE-BREW Soft-Decision Decoding Wrapper}
\label{app:decoding-wrapper}

CORE-BREW uses a score-thresholded Chase-style list-decoding wrapper. Erased positions remain zero-LLR entries for likelihood scoring, while the list decoder generates a small set of hard-decision candidates by considering alternative assignments at low-reliability positions, with erased positions prioritized and ties among equal-reliability positions resolved deterministically. A bounded-distance decoder is applied to each candidate, and a block is accepted only when at least one decoded candidate matches the designated codeword and the block score is at least $\tau_{\mathrm{blk}}$~\cite{chase1972,lin2004error,richardson2008modern}. Unlike the hard-decision ablation, CORE-BREW does not require the raw hard-decision vector itself to lie within the BCH bounded-distance radius of the designated codeword, since limited Chase perturbations may produce candidates that decode to the designated codeword. Instead, the acceptance rule combines designated-codeword verification, score thresholding, and a bounded candidate-search budget.

\FloatBarrier
\section{Theoretical Proofs}
\label{app:theory-proofs}
This appendix provides detailed proofs of the main analytical claims in Section~\ref{sec:Theory}. We first prove the Constant Hit-Rate property of the embedding rule, then analyze false positives for BREW's bounded-distance designated-codeword rule as a hard-decision reference, and finally derive score-based tail bounds for CORE-BREW's soft-decision acceptance rule. We conclude with standard concentration results for finite-sample estimation of empirical detection rates used in Section~\ref{sec:experiments}.

\subsection{Proof of Constant Hit-Rate Embedding}
\label{app:chr-proof}
We first establish the Constant Hit-Rate property underlying the calibrated channel model used in Sections~\ref{sec:method}--\ref{sec:Theory}.

\begin{proposition}[Constant hit-rate property]
Let $p^\star\in(1/2,1)$ be a fixed target hit-rate and $L_t\subseteq\mathcal{V}$ be a target token list at position $t$.
Let $\ell_t(v)$ denote the base model logit for token $v\in\mathcal{V}$ and let
\[
p_t(v) = \frac{\exp(\ell_t(v))}{\sum_{u\in\mathcal{V}}\exp(\ell_t(u))}
\]
be the corresponding softmax distribution. Define
\[
m_t = \sum_{v\in L_t} p_t(v),
\]
and choose the bias
\[
\delta_t = \log\!\left(\frac{p^\star}{1-p^\star}\right) - \log\!\left(\frac{m_t}{1-m_t}\right).
\]
Define biased logits by $\tilde{\ell}_t(v)=\ell_t(v)+\delta_t$ for $v\in L_t$ and $\tilde{\ell}_t(v)=\ell_t(v)$ otherwise,
and let $\tilde{p}_t=\mathrm{softmax}(\tilde{\ell}_t)$. Then the biased distribution satisfies
\[
\sum_{v\in L_t}\tilde{p}_t(v)=p^\star.
\]
\end{proposition}

\begin{proof}
Let
\[
A_t=\sum_{v\in L_t}\exp(\ell_t(v)),\qquad B_t=\sum_{v\notin L_t}\exp(\ell_t(v)).
\]
Then $m_t=A_t/(A_t+B_t)$ and $(1-m_t)=B_t/(A_t+B_t)$, so $A_t/B_t=m_t/(1-m_t)$.
Under the bias, the total mass on $L_t$ becomes
\[
\sum_{v\in L_t}\tilde{p}_t(v)
=\frac{\sum_{v\in L_t}\exp(\ell_t(v)+\delta_t)}{\sum_{u\in\mathcal{V}}\exp(\tilde{\ell}_t(u))}
=\frac{e^{\delta_t}A_t}{e^{\delta_t}A_t+B_t}.
\]
Substitute $e^{\delta_t}=\frac{p^\star}{1-p^\star}\cdot \frac{1-m_t}{m_t}
=\frac{p^\star}{1-p^\star}\cdot \frac{B_t}{A_t}$ to obtain
\[
\frac{e^{\delta_t}A_t}{e^{\delta_t}A_t+B_t}
=\frac{\frac{p^\star}{1-p^\star}B_t}{\frac{p^\star}{1-p^\star}B_t+B_t}
=\frac{\frac{p^\star}{1-p^\star}}{\frac{p^\star}{1-p^\star}+1}
=p^\star.
\]
\end{proof}


\paragraph{Remark (distortion guards and erasures).}
The proposition above describes positions for which the Constant Hit-Rate target is fully realized. When distortion guards such as bias clipping or embedding skips are activated, the equality need not hold at every position. Skipped positions are treated as erasures with zero LLR, while the position-wise score contributions remain bounded. Accordingly, the position-homogeneous BSC interpretation in Section~\ref{sec:Theory} applies to fully calibrated, non-erased positions.


\subsection{Hard-Decision Reference: False Positive Bounds under $H_0$}
\label{app:hard-fpr}

We derive false positive bounds for BREW's bounded-distance designated-codeword rule, which serves as the hard-decision reference in Section~\ref{sec:Theory}.
Under the null hypothesis $H_0$, the text is unwatermarked and independent of the secret key $K$.
Our analysis uses two levels of approximation.
First, for a single block, we use the standard uniform-block approximation that the hard-decision block induced by the keyed partition behaves approximately as a uniform element of $\{0,1\}^n$ under $H_0$.
Second, for text-level aggregation, we use an idealized block-level independence approximation for analytical intuition.
This aggregation model does not require exact bit-wise independence within a block.

\begin{proposition}[Single-block false positive probability]
\label{prop:single-block-fpr}
Let $C\subseteq\{0,1\}^n$ be a binary BCH code with minimum distance $d_{\min}$ and unique-decoding radius
$t=\lfloor(d_{\min}-1)/2\rfloor$~\cite{lin2004error}. 
Fix a designated codeword $c\in C$. 
Under the standard uniform-block approximation for $H_0$, the single-block false match probability is
\[
p_0 
= \Pr\!\bigl(\mathrm{dist}_H(B,c)\le t\bigr)
= \frac{V_2(n,t)}{2^n},
\qquad
V_2(n,t)=\sum_{i=0}^{t}\binom{n}{i}.
\]
\end{proposition}

\begin{proof}
Under the uniform-block approximation, $B$ is treated as uniformly distributed over $\{0,1\}^n$. The event $\mathrm{dist}_H(B,c)\le t$ holds if and only if $B$ lies in the Hamming ball of radius $t$ around $c$. That ball contains exactly $V_2(n,t)$ binary strings. Therefore, the probability is $V_2(n,t)/2^n$.
\end{proof}

\paragraph{Text-level aggregation model.}
For text-level FPR analysis, we use an idealized approximation in which block-level spurious match indicators are treated as approximately independent Bernoulli variables. 
This approximation is standard in designated-codeword analyses and is useful for understanding how text-level FPR scales with the number of complete blocks. 
However, for the fixed-partition instantiation in Appendix~\ref{app:algorithms}, the approximation is not an exact consequence of independently seeded block partitions. 
The same vocabulary partition is reused across blocks, so repeated tokens and natural-language correlations can introduce dependencies. 
Accordingly, we use the aggregate Chernoff bound below as an idealized analysis tool, complemented by empirical calibration and validation on held-out unwatermarked text.

For the shifted-window analysis below, we continue to use the same uniform-block approximation for each candidate block vector. 
This approximation is used to estimate single-block false-match probabilities and shift-overlap corrections; text-level aggregation is handled separately through the idealized block-level aggregation model described above.

For each candidate shift $s$, let
\[
E_s
=
\left\{
\mathrm{dist}_H(B^{(s)},c)\le t
\right\}
\]
denote the event that the shifted hard-decision block lies within the bounded-distance region of the designated codeword.


\begin{lemma}[Adjacent-shift joint probability under the uniform-block model]
\label{lem:joint}
Let $\Delta \in \{1, \ldots, n-1\}$ and let $c \in \{0,1\}^n$ be a fixed designated codeword. Define the codeword shift autocorrelation
\[
a(\Delta) = |\{i \in [0, n-\Delta) : c[i+\Delta] = c[i]\}|.
\]
Under an i.i.d.\ Bernoulli$(1/2)$ model for the $n+\Delta$ bits spanning the two overlapping windows,
\begin{align*}
    \Pr(E_s \cap E_{s+\Delta}) 
    = &\sum_{a=0}^{a(\Delta)} \sum_{b=0}^{n-\Delta-a(\Delta)}
    \binom{a(\Delta)}{a}\!\! \binom{n-\Delta-a(\Delta)}{b} 2^{-(n-\Delta)}
    \\&\times F_\Delta(t-a-b)\, F_\Delta\!\big(t-a-(n-\Delta-a(\Delta))+b\big),
\end{align*}
where $F_\Delta(k) = \Pr(\mathrm{Bin}(\Delta, 1/2) \leq k)$ for $k \geq 0$ and $F_\Delta(k) = 0$ for $k < 0$.
\end{lemma}

\begin{proof}
Decompose the two windows as $B^{(s)} = (X, Y)$ and $B^{(s+\Delta)} = (Y, Z)$, where $|X| = |Z| = \Delta$ and $|Y| = n - \Delta$. Under the uniform-block approximation, the non-overlapping portions $X$ and $Z$ and the shared portion $Y$ are treated as independent uniform bit strings. Define:
\begin{itemize}
\item $A \sim \mathrm{Bin}(a(\Delta), 1/2)$: errors of $Y$ at the $a(\Delta)$ 
positions where $c[i+\Delta] = c[i]$ (which contribute identically to both events);
\item $B \sim \mathrm{Bin}(n-\Delta-a(\Delta), 1/2)$: errors of $Y$ in $E_s$ at 
positions where $c[i+\Delta] \neq c[i]$ (errors in $E_{s+\Delta}$ at these 
positions equal $n-\Delta-a(\Delta) - B$);
\item $W_X \sim \mathrm{Bin}(\Delta, 1/2)$ and $W_Z \sim \mathrm{Bin}(\Delta, 1/2)$:
errors on the non-shared portions.
\end{itemize}
The total errors satisfy $d_H(B^{(s)}, c) = W_X + A + B$ and $d_H(B^{(s+\Delta)}, c) = W_Z + A + (n-\Delta-a(\Delta)) - B$.
The joint event $\{E_s \cap E_{s+\Delta}\}$ requires both sums $\leq t$. Conditioning on $(A, B)$ and using independence of $W_X, W_Z$ yields the stated closed form.
\end{proof}

\begin{proposition}[Window-shifting bound for the hard-decision reference]
\label{prop:window-shift-fpr}
Under the uniform-block approximation for candidate block vectors, the single-block false positive probability under window shifting satisfies
\[
\Pr\left(\bigcup_{s \in \mathcal{S}} E_s\right)
\leq S p_0 - (S-1)\Pr(E_s \cap E_{s+1})
\leq S p_0.
\]
\end{proposition}

\begin{proof}
The first inequality follows from Hunter's spanning-tree bound~\cite{Hun76} applied to the chain spanning tree on $\mathcal{S}$. The second follows from $\Pr(E_s \cap E_{s+1}) \geq 0$. The joint probability $\Pr(E_s \cap E_{s+1})$ admits the closed-form expression 
given in Lemma~\ref{lem:joint} with $\Delta=1$, depending on the codeword shift autocorrelation $a(1) = |\{i : c[i+1] = c[i]\}|$.
\end{proof}

\begin{proposition}[Aggregate FPR Chernoff bound under an idealized aggregation model]
Let $X_1,\ldots,X_M$ be block-level spurious match indicators under $H_0$.
Under the idealized block-level aggregation model, suppose these indicators are treated as independent Bernoulli variables satisfying
$\Pr(X_j=1)\le p_{\mathrm{blk}}$ for all $j$.
Let
\[
\rho_M=\frac{1}{M}\sum_{j=1}^M X_j.
\]
Then for any threshold $\theta\in(p_{\mathrm{blk}},1)$,
\[
\Pr(\rho_M\ge \theta)
\le
\exp\!\bigl(-M\,D(\theta\|p_{\mathrm{blk}})\bigr),
\]
where
\[
D(\theta\|p_{\mathrm{blk}})
=
\theta\log\frac{\theta}{p_{\mathrm{blk}}}
+
(1-\theta)\log\frac{1-\theta}{1-p_{\mathrm{blk}}}
\]
is the binary Kullback--Leibler divergence.
\end{proposition}

\begin{proof}
Under the idealized aggregation model, the block-level indicators are dominated by independent Bernoulli variables with success probability at most $p_{\mathrm{blk}}$. The result follows from the standard Chernoff bound for the upper tail of their empirical mean.
\end{proof}



\paragraph{Hard-decision text-level FPR.}
Applying the aggregate bound with
\[
p_{\mathrm{blk}}^{(\mathrm{shift})}
:=
S p_0-(S-1)\Pr(E_s\cap E_{s+1})
\le S p_0
\]
gives the idealized text-level FPR bound for the hard-decision reference used in Section~\ref{sec:Theory}. For the fixed-partition instantiation, this bound should be interpreted together with empirical FPR calibration and validation under $H_0$.


\subsection{Hard-Decision Reference: Single-Block Success Probability under $H_1$}
\label{app:hard-h1}

For completeness, we restate the standard bounded-distance success probability of the hard-decision reference under the alternative hypothesis $H_1$, assuming a BSC observation model.

\begin{proposition}[Single-block success probability]
Assume that under $H_1$ (no attacks) the channel from a designated codeword $c\in C$ to the detector's PRF-filled hard-decision block $\widetilde{B}\in\{0,1\}^n$ is a binary symmetric channel with crossover probability $q\in[0,1/2)$. 
Assume the bounded-distance block rule succeeds whenever the number of bit errors in $\widetilde{B}$ is at most $t$ and fails otherwise. Then the single-block success probability is
\[
p_{1,\mathrm{HD}}
=\Pr\!\bigl(\mathrm{dist}_H(\widetilde{B},c)\le t \,\big|\, H_1\bigr)
=\sum_{i=0}^{t}\binom{n}{i}q^i(1-q)^{n-i}.
\]
\end{proposition}

\begin{proof}
Under a BSC with crossover $q$ on the PRF-filled hard-decision representation, the number of bit errors $\mathrm{dist}_H(\widetilde{B},c)$ has distribution $\mathrm{Bin}(n,q)$. The bounded-distance rule succeeds if and only if the error count is at most $t$, yielding the stated sum.
\end{proof}


\subsection{CORE-BREW: Score-Based Tail Bounds}
\label{app:core-score-bounds}

We next analyze the score-based acceptance rule used by CORE-BREW. Since the block score is a sum of bounded random variables, it admits exponential tail bounds under both $H_0$ and $H_1$. These bounds characterize false positives under score-based rejection without requiring the acceptance region to coincide with a bounded-distance Hamming ball.


For the fixed-partition instantiation, the independence assumptions in this subsection should be understood as an idealized bounded-variable analysis model for score calibration rather than exact probabilistic guarantees. Let $p^\star\in(1/2,1)$ be the Constant Hit-Rate parameter and let
\[
L = \log\!\left(\frac{p^\star}{1-p^\star}\right).
\]
For each block $j$ and shift $s$, define per-bit LLRs $\Lambda_b^{(j,s)}\in[-L,L]$, where erasures correspond to $\Lambda=0$. Define the designated-codeword score
\[
S^{(j)}(s)=\sum_{b=0}^{n-1}(2c^{(j)}[b]-1)\,\Lambda_b^{(j,s)}.
\]
A necessary condition for CORE-BREW to accept a block is
\[
\max_s S^{(j)}(s)\ge \tau_{\mathrm{blk}}.
\]
Therefore, bounding $\Pr(\max_s S^{(j)}(s)\ge \tau_{\mathrm{blk}})$ provides an upper bound on the block-level false positive probability.

\begin{proposition}[Score tail bounds under $H_0$ and $H_1$]
For the following score tail bound, we use an analysis model in which, for a fixed block $j$ and shift $s$, the variables
\[
Z_b=(2c^{(j)}[b]-1)\,\Lambda_b^{(j,s)}
\]
are treated as independent and bounded, with $Z_b\in[-L,L]$.

\textbf{(Null hypothesis $H_0$.)}
If $\mathbb{E}[Z_b]=0$ for all $b$ under $H_0$, then for any $\tau_{\mathrm{blk}}>0$,
\[
\Pr_{H_0}\!\bigl(S^{(j)}(s)\ge \tau_{\mathrm{blk}}\bigr)
\le
\exp\!\left(-\frac{\tau_{\mathrm{blk}}^2}{2nL^2}\right).
\]
Consequently, with window shifting over $S=2s_{\max}+1$ shifts,
\[
\Pr_{H_0}\!\left(\max_{s} S^{(j)}(s)\ge \tau_{\mathrm{blk}}\right)
\le
S\exp\!\left(-\frac{\tau_{\mathrm{blk}}^2}{2nL^2}\right).
\]

\textbf{(Alternative hypothesis $H_1$.)}
Suppose that at the true alignment $s=s_{\mathrm{true}}$, there are $n_{\mathrm{eff}}$ non-erased positions and
the non-erased $Z_b$ satisfy $\mathbb{E}[Z_b]=\mu>0$ while erased positions contribute $Z_b=0$.
Then for any $\tau_{\mathrm{blk}} < n_{\mathrm{eff}}\mu$,
\[
\Pr_{H_1}\!\bigl(S^{(j)}(s_{\mathrm{true}})\le \tau_{\mathrm{blk}}\bigr)
\le
\exp\!\left(-\frac{(n_{\mathrm{eff}}\mu-\tau_{\mathrm{blk}})^2}{2n_{\mathrm{eff}}L^2}\right).
\]
\end{proposition}

\begin{proof}
Under this bounded-independence analysis model, both inequalities follow from Hoeffding's inequality for sums of bounded independent random variables~\cite{Hoeffding1963}. Under $H_0$, apply Hoeffding to $S^{(j)}(s)=\sum_b Z_b$ with mean $0$ and bounds $\pm L$.
The shift-max bound follows by a union bound over $S$ shifts. Under $H_1$, apply Hoeffding to the sum over the $n_{\mathrm{eff}}$ non-erased positions (erased positions contribute $0$ deterministically).
\end{proof}

\paragraph{Remark (decoding does not increase FPR beyond the score bound).}
CORE-BREW additionally requires the decoding wrapper to return the designated codeword $c^{(j)}$. This can only \emph{reduce} acceptance relative to the pure score test, so the bounds above remain valid upper bounds on block-level false positives. Conversely, the score separation under $H_1$, together with Chase-style list decoding, explains how CORE-BREW can improve detection power~\cite{chase1972,lin2004error,richardson2008modern}.

\subsection{Aggregate False Negative Bound}
We state the standard binomial lower-tail bound used to convert single-block success probabilities into a text-level FNR bound.

\begin{proposition}[Aggregate false negative bound]
Let $Y_1,\ldots,Y_M$ be block-level success indicators under $H_1$. 
Under the block-level independence analysis model, suppose that $\Pr(Y_j=1)\ge p_1$ for all $j$.
Let
\[
\rho_M=\frac{1}{M}\sum_{j=1}^M Y_j.
\]
Then for any threshold $\theta\in(0,1)$ with $\theta<p_1$,
\[
\Pr(\rho_M<\theta)\le \exp\!\bigl(-M\,D(\theta\|p_1)\bigr),
\]
where $D(\theta\|p_1)$ is the binary relative entropy.
\end{proposition}

\begin{proof}
Under the block-level independence analysis model, the lower tail of the average success rate is bounded by the standard Chernoff bound with dominating Bernoulli parameter $p_1$. This yields the stated inequality.
\end{proof}


\paragraph{Application to the detection rules.}
For the hard-decision reference, $p_1$ can be taken as $p_{1,\mathrm{HD}}$ (or its shifted variant). For CORE-BREW, $p_1$ denotes the overall block-level acceptance probability, which depends on score separation, successful designated-codeword decoding, and correct alignment in the shift search. The $H_1$ score tail bound above characterizes the score-threshold component of this probability, while the overall FNR--TPR trade-off is evaluated empirically in Section~\ref{sec:experiments}.

\subsection{Concentration of Empirical FPR (Finite-Sample Justification)}
We provide a standard concentration inequality justifying empirical estimation of (single-block or text-level) FPR. This result applies to independently sampled evaluation texts or trials and is used to justify uncertainty estimates for empirical FPR measurements.

\begin{lemma}[Concentration of empirical FPR]
Let $X_1,\ldots,X_N$ be independent Bernoulli variables with
$\Pr(X_i=1)=p_{\mathrm{FP}}$,
representing false positives for $N$ independent trials under $H_0$ (a ``trial'' may refer to a block event or a text-level event).
Let
\[
\hat{p}_{\mathrm{FP},N}
=
\frac{1}{N}\sum_{i=1}^N X_i
\]
be the empirical estimate of $p_{\mathrm{FP}}$. Then for any $\varepsilon>0$,
\[
\Pr\!\bigl(
|\hat{p}_{\mathrm{FP},N}-p_{\mathrm{FP}}|
\ge \varepsilon
\bigr)
\le
2\exp(-2N\varepsilon^2).
\]
In particular, to guarantee
$\Pr(|\hat{p}_{\mathrm{FP},N}-p_{\mathrm{FP}}|\ge\varepsilon)\le\delta$,
it suffices to take
\[
N
\ge
\frac{1}{2\varepsilon^2}
\log\!\left(\frac{2}{\delta}\right).
\]
\end{lemma}

\begin{proof}
This is Hoeffding's inequality applied to the empirical mean of bounded independent random variables~\cite{Hoeffding1963}.
\end{proof}


\section{Additional Experiments}
\label{app:additional-experiments}

This appendix provides additional experiments and diagnostics that complement the main results in Section~\ref{sec:experiments}. We focus on (i) the robustness--quality trade-off of Constant Hit-Rate embedding, (ii) the contribution of soft-decision processing relative to its hard-decision ablation, and (iii) the behavior of CORE-BREW under different operating parameters and evaluation settings.

\paragraph{Schemes compared.}
Unless otherwise stated, experiments follow the setup in Section~\ref{sec:experiments}, using the same models, prompts, sampling settings, BCH parameters, and evaluation pipeline. For a matched comparison, we use BCH $(63,7,15)$ for both BREW and CORE-BREW, although the original BREW evaluation uses BCH $(31,6,7)$. This controls for differences in ECC strength and isolates the contribution of CORE-BREW's embedding and decoding mechanisms. We consider:
\begin{itemize}
  \item \textbf{BREW~\cite{kim2026blockwise}:} the designated-codeword baseline using fixed-bias embedding, hard-decision BCH decoding, and window-shifting detection.
  \item \textbf{CORE-BREW:} our proposed method, combining Constant Hit-Rate embedding, entropy-aware erasures, calibrated LLR scoring, and Chase-style list decoding with score-based acceptance~\cite{chase1972,lin2004error,richardson2008modern}.
  \item \textbf{CORE-BREW w/o Soft Decoding:} a hard-decision ablation that retains the same Constant Hit-Rate embedding and erasure decisions as CORE-BREW, but replaces LLR-based scoring and list decoding with bounded-distance designated-codeword verification using PRF-filled erased positions. This ablation provides a matched hard-decision reference for evaluating the contribution of soft-decision processing.
\end{itemize}

\paragraph{Entropy-aware embedding and erasures.}
CORE-BREW uses entropy-aware erasures to limit excessive distributional perturbation in low-entropy contexts. After computing and clipping the Constant Hit-Rate bias, a position is erased when the base distribution is too peaked ($\max_v p_t(v)\ge q_{\max}$), the entropy falls below $H_{\min}$, or the post-bias target-list mass does not exceed $1/2$. Erased positions are sampled without watermark bias and assigned zero LLR during model-aware detection, allowing the decoder to ignore unreliable evidence.

\subsection{Clean Detection Results}
\label{app:clean-breakdown}

\begin{table}[t]
\centering
\caption{Clean detection performance on Mistral-7B. We report mean $\pm$ standard deviation over multiple random seeds at each method's default or recommended detection threshold.}
\label{tab:clean-mistral}
\small
\setlength{\tabcolsep}{4pt}
\begin{tabular}{lcccc}
\toprule
& \multicolumn{2}{c}{C4} & \multicolumn{2}{c}{OpenGen} \\
\cmidrule(lr){2-3} \cmidrule(lr){4-5}
Method
& TPR $\uparrow$
& FPR $\downarrow$
& TPR $\uparrow$
& FPR $\downarrow$ \\
\midrule
MPAC
& $1.0000 \pm 0.0000$
& $0.7333 \pm 0.0062$
& $1.0000 \pm 0.0000$
& $0.8683 \pm 0.0094$ \\

Qu et al.~\cite{qu2025provably}
& $1.0000 \pm 0.0000$
& $0.9250 \pm 0.0071$
& $1.0000 \pm 0.0000$
& $0.9383 \pm 0.0131$ \\

BREW
& $0.9267 \pm 0.0184$
& $0.0000 \pm 0.0000$
& $0.9717 \pm 0.0085$
& $0.0000 \pm 0.0000$ \\

CORE-BREW
& $0.9167 \pm 0.0047$
& $0.0000 \pm 0.0000$
& $0.9667 \pm 0.0024$
& $0.0000 \pm 0.0000$ \\
\bottomrule
\end{tabular}
\end{table}

Table~\ref{tab:clean-mistral} summarizes clean detection performance on Mistral-7B at text length $T=500$, using each method's default or recommended detection threshold. MPAC~\cite{yoo2024advancing} and Qu et al.~\cite{qu2025provably} achieve near-perfect TPR but exhibit substantially higher FPR, indicating limited discrimination between watermarked and unwatermarked text. In contrast, BREW~\cite{kim2026blockwise} and CORE-BREW maintain zero observed FPR while achieving high clean-detection TPR on both C4 and OpenGen. Their clean detection performance is comparable, with only small differences in TPR across the two datasets. CORE-BREW uses the default operating point $p^\star=0.9$, with the full hit-rate sweep reported in Appendix~\ref{app:hit-rate-sweep}.

\subsection{Effect of Target Hit-Rate $p^\star$ on Detection and Payload Recovery}
\label{app:hit-rate-sweep}

\begin{figure}[t]
  \begin{center}
    \includegraphics[width=\columnwidth]{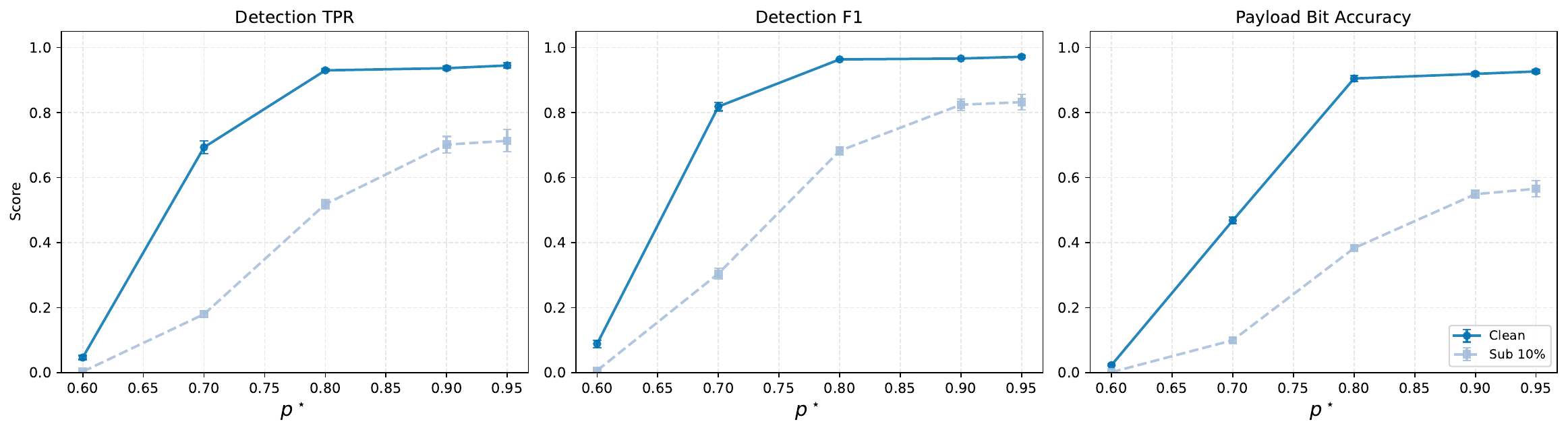}
    \caption{Effect of the target hit-rate $p^\star$ on CORE-BREW using OPT-1.3B with C4 under clean generation and 10\% substitution. We report detection TPR, detection F1, and payload bit accuracy. Increasing $p^\star$ improves both detection and payload recovery, with gains largely saturating at high hit rates.}
    \label{fig:hit_rate_sweep}
  \end{center}
\end{figure}

We analyze the effect of the target hit-rate $p^\star$ on CORE-BREW. Figure~\ref{fig:hit_rate_sweep} reports TPR, F1, and payload bit accuracy under clean generation and 10\% substitution. Under clean generation, all three metrics improve sharply from $p^\star=0.6$ to $0.8$ and largely saturate thereafter. In particular, TPR and F1 exceed 90\% from $p^\star=0.8$ onward, while payload bit accuracy reaches approximately 90\% at $p^\star=0.8$ with only marginal gains at higher values.

Under 10\% substitution, performance is lower overall but follows the same trend, with TPR, F1, and payload bit accuracy increasing markedly up to $p^\star=0.9$. The additional gain from $p^\star=0.9$ to $0.95$ is relatively small. We therefore use $p^\star=0.9$ as the default setting, since clean performance is already near saturation while further increases provide only limited improvement under substitution.

\subsection{Code Length and Error-Correction Parameters}
\label{app:code-length}

\begin{table}[t]
  \caption{Code-length sensitivity of CORE-BREW on OPT-1.3B with C4 at $p^\star=0.9$ and $T=500$. We report TPR under the constraint $\mathrm{FPR}\leq 1\%$.}
  \label{tab:code_length}
  \centering
  \setlength{\tabcolsep}{5pt}
  \renewcommand{\arraystretch}{1.1}
  \begin{tabular}{lcc}
    \toprule
    \textbf{BCH $(n,k,t)$}
    & \textbf{Clean}
    & \textbf{10\% Substitution} \\
    \midrule
    $(15,5,3)$
    & $0.9350 \pm 0.0108$
    & $0.5317 \pm 0.0170$ \\

    $(31,6,7)$
    & $0.9283 \pm 0.0024$
    & $0.6400 \pm 0.0255$ \\

    $(63,7,15)$
    & $\mathbf{0.9383 \pm 0.0062}$
    & $\mathbf{0.7017 \pm 0.0249}$ \\
    \bottomrule
  \end{tabular}
\end{table}

Table~\ref{tab:code_length} compares three BCH configurations for CORE-BREW on OPT-1.3B with C4 under the same low-FPR constraint of $\mathrm{FPR}\leq 1\%$. Rather than comparing TPR at a single detection threshold, we evaluate the achievable TPR under the same false-positive requirement to compare detection power across code configurations.
Under clean generation, all three configurations achieve comparable detection performance, with $(63,7,15)$ attaining the highest TPR of $0.9383$. The advantage of the longer code becomes more pronounced under 10\% substitution, where TPR increases from $0.5317$ for $(15,5,3)$ and $0.6400$ for $(31,6,7)$ to $0.7017$ for $(63,7,15)$. Thus, the longer BCH configuration provides greater robustness to substitution while maintaining the same low-FPR requirement.
Based on these results, we use BCH $(63,7,15)$ as the default configuration when sufficient text is available. In the insertion/deletion experiments, we additionally use $s_{\max}=10$ to improve alignment robustness, as analyzed in Appendix~\ref{app:Window_shift}. For shorter continuations that cannot accommodate a complete length-63 block, CORE-BREW can instead use shorter BCH configurations, as evaluated in Appendix~\ref{app:short_text}.

\subsection{Window-Shift Sensitivity under Insertion and Deletion Attacks}
\label{app:Window_shift}

\begin{figure}[t]
  \begin{center}
    \includegraphics[width=\columnwidth]{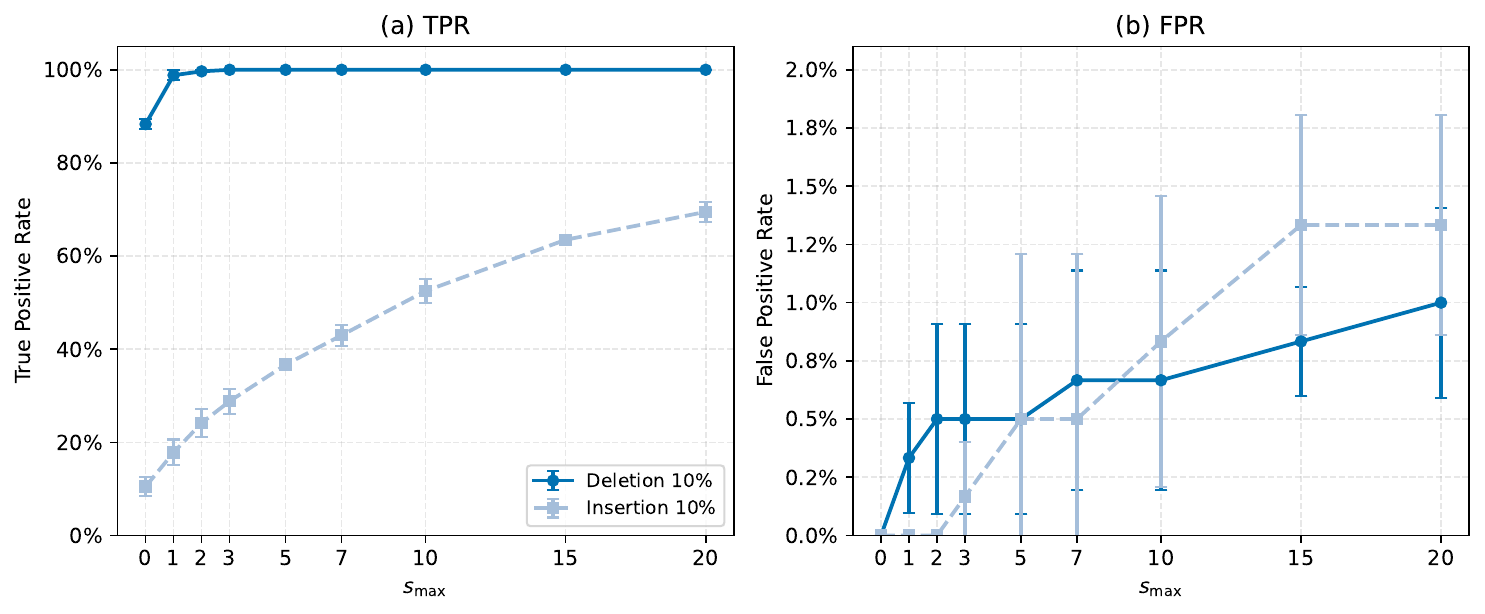}
    \caption{Sensitivity to the window-shift budget $s_{\max}$ under 10\% insertion and deletion attacks on OPT-1.3B with C4 at $p^\star=0.9$. Increasing $s_{\max}$ improves TPR, particularly under insertion, with a modest increase in observed FPR.}
    \label{fig:S_max_Performance_Analysis}
  \end{center}
\end{figure}

We analyze the sensitivity of CORE-BREW to the window-shift budget $s_{\max}$ under 10\% insertion and deletion attacks. Figure~\ref{fig:S_max_Performance_Analysis} reports TPR and FPR for $s_{\max}\in\{0,1,2,3,5,7,10,15,20\}$. Under deletion, TPR is already high at $s_{\max}=0$ ($0.8833$), rises to $0.9883$ at $s_{\max}=1$, and reaches $1.0000$ by $s_{\max}=3$, after which it saturates. In contrast, insertion is substantially more sensitive to the shift budget: TPR increases steadily from $0.1050$ at $s_{\max}=0$ to $0.5250$ at $s_{\max}=10$ and $0.6950$ at $s_{\max}=20$. Thus, deletion requires only a small alignment search, whereas insertion continues to benefit from a larger shift budget.

A larger $s_{\max}$ also increases the likelihood of spurious matches. Under deletion, the observed FPR increases from $0$ at $s_{\max}=0$ to $0.0067$ at $s_{\max}=10$ and $0.0100$ at $s_{\max}=20$. Under insertion, it remains near zero at small shift budgets, reaches $0.0083$ at $s_{\max}=10$, and $0.0133$ at $s_{\max}\in\{15,20\}$. We therefore use $s_{\max}=10$ as the default setting, which substantially improves insertion TPR over no shift search while keeping the observed FPR below $1\%$ for both attacks.

\subsection{Sensitivity to the Text-Level Match Threshold}
\label{app:match-threshold}

\begin{table}[t]
\centering
\caption{Sensitivity to the text-level match threshold under 10\% insertion on OPT-1.3B with C4 ($T=500$).}
\label{tab:match_threshold}
\small
\begin{tabular}{c cc}
\toprule
\textbf{Required matched blocks} & \textbf{TPR} & \textbf{FPR} \\
\midrule
$M_{\mathrm{match}}\ge1$ & $0.5250 \pm 0.0255$ & $0.0083 \pm 0.0062$ \\
$M_{\mathrm{match}}\ge2$ & $0.0567 \pm 0.0131$ & $0.0000 \pm 0.0000$ \\
\bottomrule
\end{tabular}
\end{table}

We examine the sensitivity to the text-level match threshold under the 10\% insertion attack, where alignment errors make block recovery particularly challenging. Increasing the requirement from one to two matched blocks reduces the observed FPR from $0.0083$ to zero, but also sharply decreases TPR from $0.5250$ to $0.0567$. Thus, requiring two matches provides only a small additional reduction in observed FPR at a substantial cost in detection power. We therefore use $M_{\mathrm{match}}\ge1$ as the default criterion.

\subsection{Text Quality Evaluation}
\label{app:text-quality-sweep}

\begin{table}[t]
\centering
\caption{Text quality at $T=500$. PPL is evaluated with a Llama-2-7B oracle, while BLEU and BERTScore-F1 are evaluated against WMT16 references. We report mean $\pm$ standard deviation over multiple runs.}
\label{tab:text_quality}
\small
\setlength{\tabcolsep}{3.5pt}
\begin{tabular}{l c c c c}
\toprule
Method & Parameter & PPL $\downarrow$ & BLEU $\uparrow$ & BERTScore-F1 $\uparrow$ \\
\midrule
Unwatermarked & -- & $8.454 \pm 0.053$ & $31.81 \pm 0.00$ & $0.881 \pm 0.000$ \\
MPAC & $\delta=3$ & $12.684 \pm 0.114$ & $22.30 \pm 0.00$ & $0.855 \pm 0.000$ \\
Qu et al.~\cite{qu2025provably} & $\delta=3$ & $13.326 \pm 0.004$ & $21.41 \pm 0.00$ & $0.847 \pm 0.000$ \\
BREW & $\delta=3$ & $13.643 \pm 0.071$ & $21.73 \pm 0.74$ & $0.853 \pm 0.001$ \\
CORE-BREW & $p^\star=0.9$ & $15.306 \pm 0.210$ & $31.82 \pm 0.00$ & $0.881 \pm 0.000$ \\
\bottomrule
\end{tabular}
\end{table}

We evaluate text quality at the default operating point of each watermarking method. Baseline methods use a fixed logit bias $\delta=3$, while CORE-BREW uses the target hit rate $p^\star=0.9$ with entropy-aware safeguards. Conditional PPL is evaluated on OPT-1.3B generations over C4 using Llama-2-7B as an independent oracle conditioned on the prompt. BLEU~\cite{Papineni2002} and BERTScore-F1~\cite{zhang2020bertscore} are evaluated on WMT16 DE$\rightarrow$EN validation translations generated by NLLB-600M against the corresponding reference translations.

CORE-BREW exhibits higher conditional PPL than the evaluated baselines, indicating lower likelihood under the independent oracle model. On WMT16, however, CORE-BREW achieves a BLEU score of $31.82$, nearly identical to $31.81$ for unwatermarked generation, and a BERTScore-F1 of $0.881$, matching that of unwatermarked generation. In contrast, the evaluated baselines show substantially lower BLEU and BERTScore-F1 under their watermarking settings. These results suggest a likelihood--robustness trade-off while maintaining reference-based lexical and semantic quality comparable to unwatermarked generation.

\paragraph{Role of Entropy-Aware Safeguards.}
An unconstrained constant hit-rate formulation can induce excessively large instantaneous biases $|\delta_t|$ when the target-list mass $m_t$ approaches $0$ or $1$, a condition frequently observed in low-entropy contexts. Such extreme biasing can distort the output distribution and degrade quality. The entropy-aware mechanism bounds $|\delta_t|$ and converts these positions into erasures, mitigating distributional distortion while preserving detectability through erasure-tolerant decoding~\cite{lin2004error,richardson2008modern}.

\subsection{Ablation on Entropy-Aware Erasures} 
\label{app:erasure-ablation}

\begin{figure}[t]
  \begin{center}
    \includegraphics[width=\columnwidth]{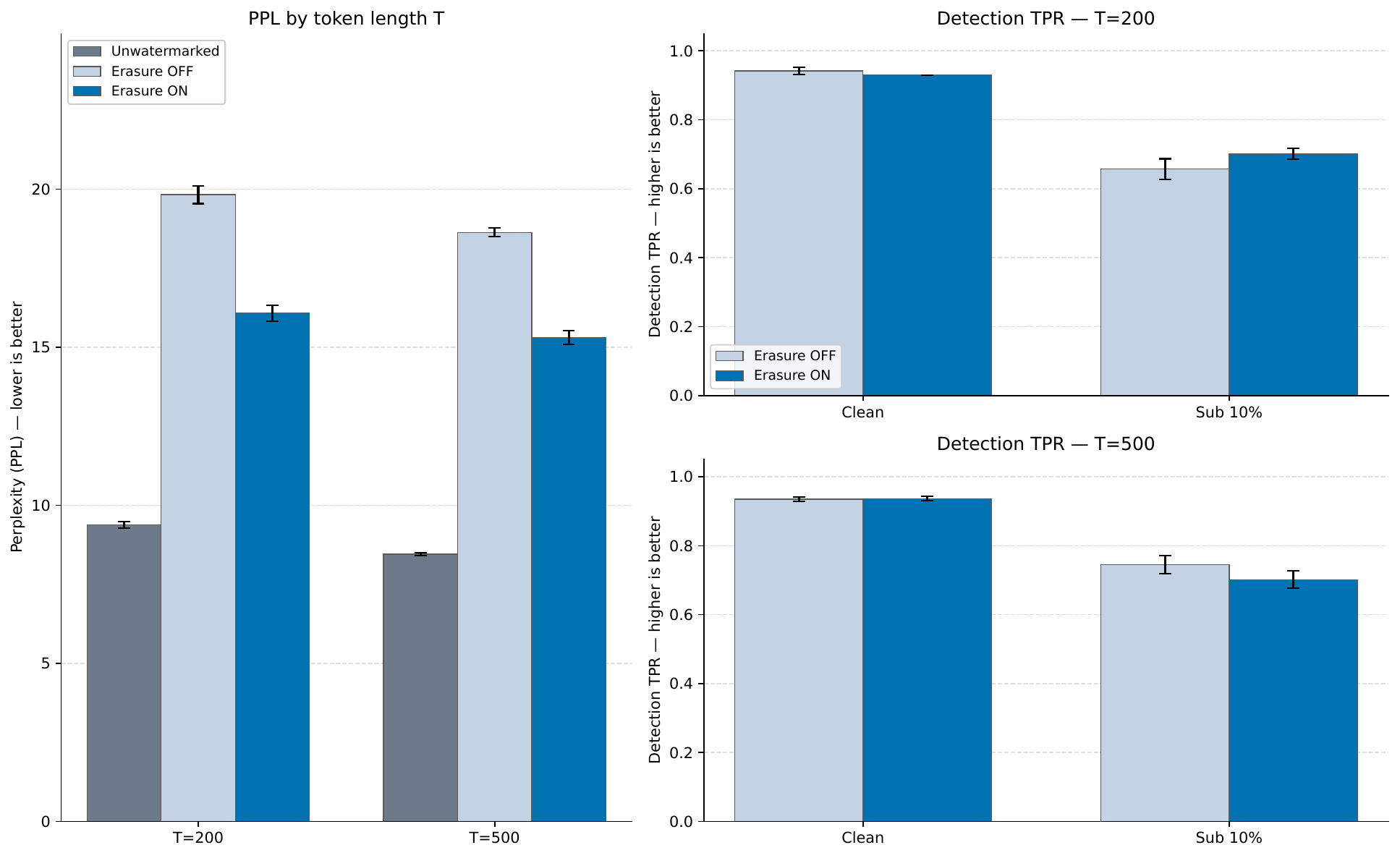}
    \caption{Effect of entropy-aware erasures on conditional PPL and detection TPR on OPT-1.3B with C4 at $p^\star=0.9$. Enabling erasures substantially reduces PPL at both $T=200$ and $T=500$, while TPR remains broadly comparable under clean and 10\% substitution settings.}
    \label{fig:erasure_ablation_p09}
  \end{center}
\end{figure}

To evaluate the system-level effect of entropy-aware erasures, Figure~\ref{fig:erasure_ablation_p09} compares CORE-BREW with the safeguard enabled and disabled on OPT-1.3B with C4 at $p^\star=0.9$. Enabling erasures substantially reduces conditional PPL, from $19.825 \pm 0.285$ to $16.072 \pm 0.255$ at $T=200$ and from $18.631 \pm 0.140$ to $15.306 \pm 0.210$ at $T=500$. Although these values remain higher than the corresponding unwatermarked PPLs of $9.376 \pm 0.104$ and $8.454 \pm 0.053$, the erasure-enabled configuration consistently exhibits lower distributional distortion than the no-erasure variant.
The corresponding changes in detection performance are relatively modest. Under clean generation, TPR remains nearly unchanged at both text lengths. Under 10\% substitution, enabling erasures improves TPR at $T=200$ but reduces it at $T=500$, indicating that the effect on detection depends on the evaluation setting rather than being uniformly beneficial.
Because the erasure mechanism jointly affects embedding and detection, this ablation should be interpreted as a system-level comparison rather than as isolating a single operation. When enabled, safeguarded positions are left unbiased during generation and contribute zero LLR at detection; when disabled, CHR embedding and constant-magnitude LLRs are applied at all positions. Overall, the erasure-enabled configuration substantially lowers conditional PPL while producing only limited changes in TPR, supporting its use as a distortion-control safeguard in CORE-BREW.

\subsection{Payload Scalability via Block-Wise Embedding}
\label{app:payload_scaling}

\begin{table}[t]
\centering
\caption{Payload scaling of CORE-BREW under clean generation using OPT-1.3B with C4. We evaluate total payload sizes from 7 to 49 bits within a fixed $T=500$ generation budget and report detection TPR, FPR, and payload bit accuracy.}
\label{tab:payload_scaling}
\small
\setlength{\tabcolsep}{6pt}
\renewcommand{\arraystretch}{1.05}

\begin{tabular}{c c c c}
\toprule
\textbf{Payload Bits}
& \textbf{TPR}
& \textbf{FPR}
& \textbf{Bit Acc.} \\
\midrule
7  & $0.9317 \pm 0.0062$ & $0.0000 \pm 0.0000$ & $0.9317 \pm 0.0062$ \\
14 & $0.9383 \pm 0.0062$ & $0.0000 \pm 0.0000$ & $0.9308 \pm 0.0042$ \\
21 & $0.9300 \pm 0.0041$ & $0.0000 \pm 0.0000$ & $0.8967 \pm 0.0152$ \\
28 & $0.9350 \pm 0.0082$ & $0.0000 \pm 0.0000$ & $0.8879 \pm 0.0136$ \\
49 & $0.9450 \pm 0.0071$ & $0.0017 \pm 0.0024$ & $0.8505 \pm 0.0067$ \\
\bottomrule
\end{tabular}
\end{table}

The payload capacity of CORE-BREW is not fundamentally constrained to 7 bits. In BCH $(63,7,15)$, $k=7$ represents the message dimension of an individual BCH codeword rather than the cumulative text-level capacity. Since CORE-BREW operates block-wise, distinct 7-bit payload chunks can be assigned across successive blocks. A $T=500$ continuation can accommodate up to seven complete length-63 blocks, corresponding to 49 payload bits when each block carries a distinct chunk. Table~\ref{tab:payload_scaling} evaluates this block-wise scalability on OPT-1.3B with C4 under clean generation.

Across payload sizes from 7 to 49 bits, text-level detection remains stable: TPR ranges from $0.9300$ to $0.9450$, while the observed FPR remains at or below $0.0017$. Payload bit accuracy decreases gradually from $0.9317$ at 7 bits to $0.8505$ at 49 bits as less repeated block-level evidence is available for each distinct payload chunk. These results show that the per-block value $k=7$ does not impose a 7-bit limit on CORE-BREW's total payload, and that payloads up to 49 bits can be embedded within $T=500$ while retaining comparable clean detection performance. Appendix~\ref{app:short_text} considers the complementary short-text setting.

\subsection{Short-Text Detection and Payload Recovery}
\label{app:short_text}

\begin{table*}[t]
  \caption{Short-text detection and payload recovery under clean generation using OPT-1.3B with C4. We report TPR, FPR, and payload bit accuracy for continuations of $T\in\{50,100\}$ tokens. ``--'' indicates that no complete BCH block can be formed at the given text length.}
  \label{tab:short_text}
  \small
  \centering
  \setlength{\tabcolsep}{5pt}
  \renewcommand{\arraystretch}{1.05}
  \begin{tabular}{l ccc ccc}
    \toprule
    & \multicolumn{3}{c}{$T=50$}
    & \multicolumn{3}{c}{$T=100$} \\
    \cmidrule(lr){2-4}\cmidrule(lr){5-7}
    \textbf{Method}
    & \textbf{TPR} & \textbf{FPR} & \textbf{Payload Bit Acc.}
    & \textbf{TPR} & \textbf{FPR} & \textbf{Payload Bit Acc.} \\
    \midrule

    MPAC ($\delta=3$)
    & $0.997$ & $0.323$ & $0.867$
    & $1.000$ & $0.368$ & $0.960$ \\

    Qu et al.~\cite{qu2025provably} ($\delta=3$)
    & $0.938$ & $0.937$ & $0.060$
    & $0.975$ & $0.932$ & $0.575$ \\
    
    \midrule

    CORE-BREW $(15,5,3)$
    & $0.930$ & $0.097$ & $0.896$
    & $0.942$ & $0.182$ & $0.937$ \\

    CORE-BREW $(31,6,7)$
    & $0.910$ & $0.002$ & $0.910$
    & $0.930$ & $0.012$ & $0.912$ \\

    CORE-BREW $(63,7,15)$
    & -- & -- & --
    & $0.912$ & $0.000$ & $0.912$ \\

    \bottomrule
  \end{tabular}
\end{table*}

While BCH $(63,7,15)$ is preferred when sufficient text is available, CORE-BREW can support shorter continuations by selecting a BCH configuration whose block length fits the available text. Table~\ref{tab:short_text} evaluates this flexibility on OPT-1.3B with C4 for $T=50$ and $T=100$. At $T=50$, BCH $(63,7,15)$ cannot form a complete block, whereas shorter BCH configurations remain applicable. In particular, BCH $(31,6,7)$ achieves a TPR of $0.910$ with an FPR of $0.002$ and a payload bit accuracy of $0.910$. However, the shorter BCH $(15,5,3)$ exhibits substantially higher FPR, reaching $0.097$ at $T=50$ and $0.182$ at $T=100$, indicating that shorter block lengths can weaken false-positive control under window shifting.

At $T=100$, CORE-BREW can use the default BCH $(63,7,15)$ configuration, allowing one complete block that carries a 7-bit payload while achieving a TPR of $0.912$, a payload bit accuracy of $0.912$, and zero observed false positives. Shorter BCH configurations can fit multiple blocks within the same continuation, but this increased flexibility comes with a trade-off in false-positive control. Thus, adapting the code length enables CORE-BREW to operate on shorter generations, while the BCH configuration should be selected by balancing text length, payload capacity, and false-positive behavior.

\subsection{Additional Results on Mistral-7B}
\label{app:mistral-results}

\begin{figure*}[t]
    \centering
    \includegraphics[width=\textwidth]{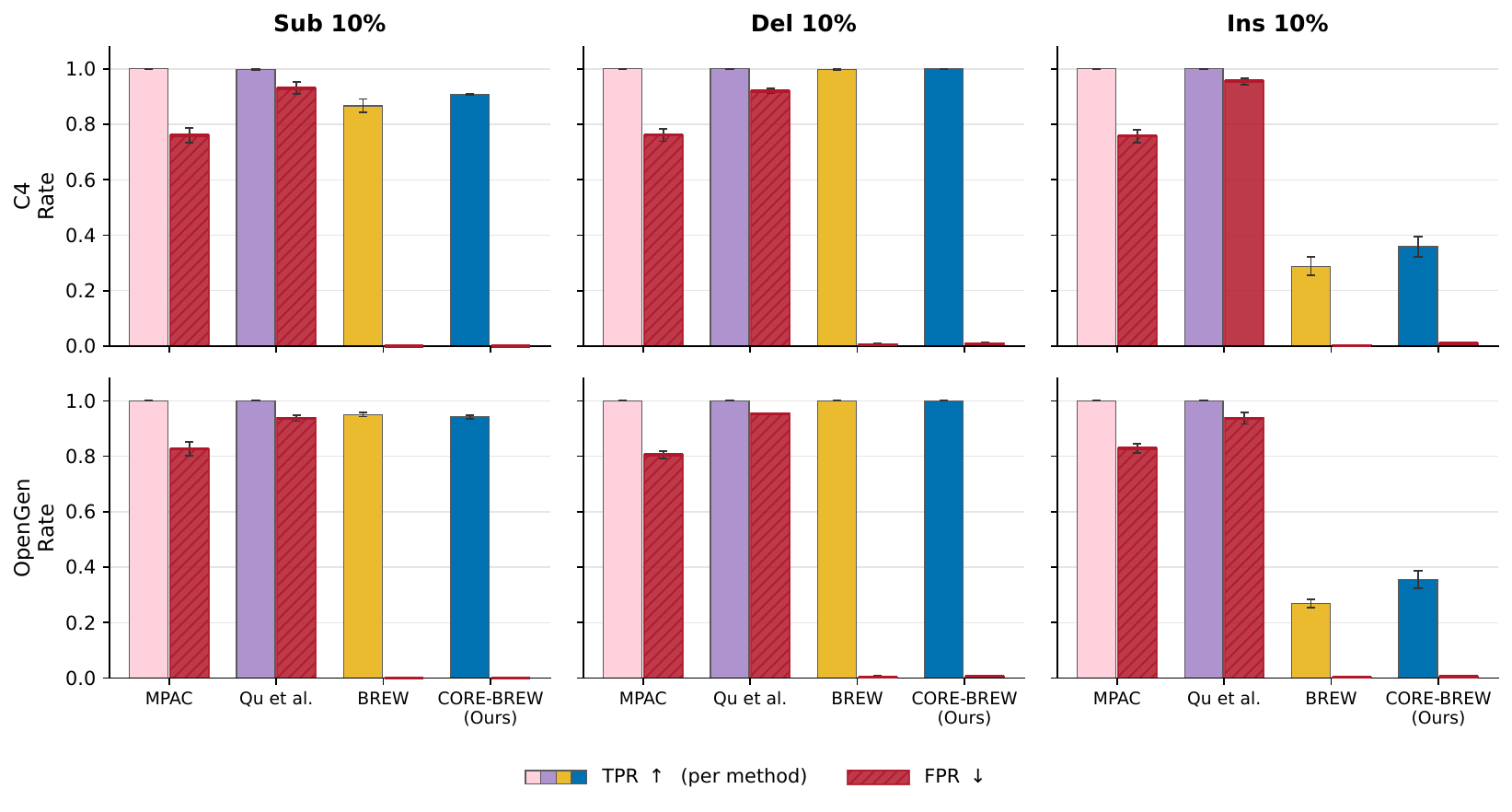}
    \caption{
    Detection performance under 10\% substitution, deletion, and insertion attacks on Mistral-7B with C4 and OpenGen. Colored bars denote TPR and red bars denote FPR. MPAC and Qu et al.~\cite{qu2025provably} retain high TPR but exhibit substantially elevated FPR, whereas BREW~\cite{kim2026blockwise} and CORE-BREW maintain low observed FPR.}
    \label{fig:mistral_attacks}
\end{figure*}

Figure~\ref{fig:mistral_attacks} provides additional robustness results on Mistral-7B using the same evaluation protocol as Section~\ref{sec:synthetic-attacks}. The overall trends are consistent with those observed on OPT-1.3B: MPAC and Qu et al.~\cite{qu2025provably} achieve high TPR but incur high FPR, while BREW~\cite{kim2026blockwise} and CORE-BREW operate at substantially lower false-positive rates.
Among the low-FPR methods, deletion is the least disruptive, with both BREW and CORE-BREW achieving near-perfect detection. Substitution yields strong detection performance, although the relative advantage between BREW and CORE-BREW varies by dataset. Insertion remains the most challenging attack, where CORE-BREW improves TPR over BREW on both C4 and OpenGen while maintaining low observed FPR.

\subsection{Detailed Results under Synthetic Token-Level Attacks}
\label{Appendix:Synthetic_Token_Level_Attacks}

\begin{table*}[t]
\centering
\caption{Detection and payload recovery performance under a 10\% substitution attack. We report TPR, FPR, and payload bit accuracy on OPT-1.3B and Mistral-7B with C4 and OpenGen as mean $\pm$ standard deviation over multiple runs.}
\label{tab:substitution_results}
\small
\setlength{\tabcolsep}{5pt}
\renewcommand{\arraystretch}{1.05}

\resizebox{\textwidth}{!}{%
\begin{tabular}{l l l c c c}
\toprule
\textbf{Model} & \textbf{Dataset} & \textbf{Method}
& \textbf{TPR} & \textbf{FPR} & \textbf{Payload Bit Acc.} \\
\midrule

\multirow{8}{*}{OPT-1.3B}
& \multirow{4}{*}{C4}
& MPAC & $1.0000 \pm 0.0000$ & $0.8317 \pm 0.0184$ & $0.9921 \pm 0.0006$ \\
& & Qu et al.~\cite{qu2025provably} & $1.0000 \pm 0.0000$ & $0.9483 \pm 0.0024$ & $1.0000 \pm 0.0000$ \\
& & BREW & $0.6550 \pm 0.0367$ & $0.0000 \pm 0.0000$ & $0.4892 \pm 0.0255$ \\
& & CORE-BREW & $0.7017 \pm 0.0249$ & $0.0000 \pm 0.0000$ & $0.5492 \pm 0.0120$ \\
\cmidrule(lr){2-6}
& \multirow{4}{*}{OpenGen}
& MPAC & $1.0000 \pm 0.0000$ & $0.9083 \pm 0.0024$ & $0.9953 \pm 0.0009$ \\
& & Qu et al.~\cite{qu2025provably} & $1.0000 \pm 0.0000$ & $0.9450 \pm 0.0141$ & $1.0000 \pm 0.0000$ \\
& & BREW & $0.6583 \pm 0.0232$ & $0.0000 \pm 0.0000$ & $0.4708 \pm 0.0228$ \\
& & CORE-BREW & $0.7217 \pm 0.0246$ & $0.0000 \pm 0.0000$ & $0.5817 \pm 0.0174$ \\

\midrule

\multirow{8}{*}{Mistral-7B}
& \multirow{4}{*}{C4}
& MPAC & $1.0000 \pm 0.0000$ & $0.7600 \pm 0.0268$ & $0.9987 \pm 0.0000$ \\
& & Qu et al.~\cite{qu2025provably} & $0.9983 \pm 0.0024$ & $0.9300 \pm 0.0216$ & $0.9933 \pm 0.0024$ \\
& & BREW & $0.8667 \pm 0.0246$ & $0.0000 \pm 0.0000$ & $0.7933 \pm 0.0143$ \\
& & CORE-BREW & $0.9083 \pm 0.0024$ & $0.0000 \pm 0.0000$ & $0.8733 \pm 0.0012$ \\
\cmidrule(lr){2-6}
& \multirow{4}{*}{OpenGen}
& MPAC & $1.0000 \pm 0.0000$ & $0.8267 \pm 0.0249$ & $0.9956 \pm 0.0013$ \\
& & Qu et al.~\cite{qu2025provably} & $1.0000 \pm 0.0000$ & $0.9367 \pm 0.0103$ & $0.9983 \pm 0.0024$ \\
& & BREW & $0.9500 \pm 0.0071$ & $0.0000 \pm 0.0000$ & $0.8908 \pm 0.0096$ \\
& & CORE-BREW & $0.9417 \pm 0.0062$ & $0.0000 \pm 0.0000$ & $0.9075 \pm 0.0074$ \\ 

\bottomrule
\end{tabular}%
}
\end{table*}


\begin{table*}[t]
\centering
\caption{Detection and payload recovery performance under a 10\% deletion attack. We report TPR, FPR, and payload bit accuracy on OPT-1.3B and Mistral-7B with C4 and OpenGen as mean $\pm$ standard deviation over multiple runs.}
\label{tab:deletion_results}
\small
\setlength{\tabcolsep}{5pt}
\renewcommand{\arraystretch}{1.05}

\resizebox{\textwidth}{!}{%
\begin{tabular}{l l l c c c}
\toprule
\textbf{Model} & \textbf{Dataset} & \textbf{Method}
& \textbf{TPR} & \textbf{FPR} & \textbf{Payload Bit Acc.} \\
\midrule

\multirow{8}{*}{OPT-1.3B}
& \multirow{4}{*}{C4}
& MPAC & $1.0000 \pm 0.0000$ & $0.8450 \pm 0.0187$ & $0.9945 \pm 0.0012$ \\
& & Qu et al.~\cite{qu2025provably} & $1.0000 \pm 0.0000$ & $0.9283 \pm 0.0131$ & $1.0000 \pm 0.0000$ \\
& & BREW & $0.9983 \pm 0.0024$ & $0.0033 \pm 0.0024$ & $0.9900 \pm 0.0061$ \\
& & CORE-BREW & $1.0000 \pm 0.0000$ & $0.0067 \pm 0.0047$ & $1.0000 \pm 0.0000$ \\ 
\cmidrule(lr){2-6}
& \multirow{4}{*}{OpenGen}
& MPAC & $1.0000 \pm 0.0000$ & $0.8833 \pm 0.0201$ & $0.9941 \pm 0.0009$ \\
& & Qu et al.~\cite{qu2025provably} & $1.0000 \pm 0.0000$ & $0.9433 \pm 0.0062$ & $1.0000 \pm 0.0000$ \\
& & BREW & $1.0000 \pm 0.0000$ & $0.0017 \pm 0.0024$ & $0.9883 \pm 0.0031$ \\
& & CORE-BREW & $1.0000 \pm 0.0000$ & $0.0050 \pm 0.0041$ & $0.9992 \pm 0.0012$ \\

\midrule

\multirow{8}{*}{Mistral-7B}
& \multirow{4}{*}{C4}
& MPAC & $1.0000 \pm 0.0000$ & $0.7617 \pm 0.0232$ & $0.9990 \pm 0.0008$ \\
& & Qu et al.~\cite{qu2025provably} & $1.0000 \pm 0.0000$ & $0.9200 \pm 0.0082$ & $1.0000 \pm 0.0000$ \\
& & BREW & $0.9983 \pm 0.0024$ & $0.0050 \pm 0.0041$ & $0.9767 \pm 0.0062$ \\
& & CORE-BREW & $1.0000 \pm 0.0000$ & $0.0083 \pm 0.0062$ & $0.9967 \pm 0.0047$ \\
\cmidrule(lr){2-6}
& \multirow{4}{*}{OpenGen}
& MPAC & $1.0000 \pm 0.0000$ & $0.8050 \pm 0.0122$ & $0.9971 \pm 0.0004$ \\
& & Qu et al.~\cite{qu2025provably} & $1.0000 \pm 0.0000$ & $0.9533 \pm 0.0024$ & $0.9983 \pm 0.0024$ \\
& & BREW & $1.0000 \pm 0.0000$ & $0.0033 \pm 0.0047$ & $0.9767 \pm 0.0012$ \\
& & CORE-BREW & $1.0000 \pm 0.0000$ & $0.0067 \pm 0.0024$ & $0.9958 \pm 0.0024$ \\

\bottomrule
\end{tabular}%
}
\end{table*}

\begin{table*}[t]
\centering
\caption{Detection and payload recovery performance under a 10\% insertion attack. We report TPR, FPR, and payload bit accuracy on OPT-1.3B and Mistral-7B with C4 and OpenGen as mean $\pm$ standard deviation over multiple runs.}
\label{tab:insertion_results}
\small
\setlength{\tabcolsep}{5pt}
\renewcommand{\arraystretch}{1.05}

\resizebox{\textwidth}{!}{%
\begin{tabular}{l l l c c c}
\toprule
\textbf{Model} & \textbf{Dataset} & \textbf{Method}
& \textbf{TPR} & \textbf{FPR} & \textbf{Payload Bit Acc.} \\
\midrule

\multirow{8}{*}{OPT-1.3B}
& \multirow{4}{*}{C4}
& MPAC & $1.0000 \pm 0.0000$ & $0.8367 \pm 0.0085$ & $0.9912 \pm 0.0007$ \\
& & Qu et al.~\cite{qu2025provably} & $1.0000 \pm 0.0000$ & $0.9417 \pm 0.0125$ & $1.0000 \pm 0.0000$ \\
& & BREW & $0.4567 \pm 0.0125$ & $0.0000 \pm 0.0000$ & $0.2692 \pm 0.0133$ \\
& & CORE-BREW & $0.5250\pm 0.0255$ & $0.0083 \pm 0.0062$ & $0.3292 \pm 0.0204$ \\ 
\cmidrule(lr){2-6}
& \multirow{4}{*}{OpenGen}
& MPAC & $1.0000 \pm 0.0000$ & $0.8767 \pm 0.0262$ & $0.9933 \pm 0.0008$ \\
& & Qu et al.~\cite{qu2025provably} & $1.0000 \pm 0.0000$ & $0.9550 \pm 0.0108$ & $1.0000 \pm 0.0000$ \\
& & BREW & $0.4550 \pm 0.0389$ & $0.0033 \pm 0.0047$ & $0.2592 \pm 0.0283$ \\
& & CORE-BREW & $0.5250 \pm 0.0122$ & $0.0100 \pm 0.0000$ & $0.3292 \pm 0.0176$ \\

\midrule

\multirow{8}{*}{Mistral-7B}
& \multirow{4}{*}{C4}
& MPAC & $1.0000 \pm 0.0000$ & $0.7583 \pm 0.0232$ & $0.9968 \pm 0.0009$ \\
& & Qu et al.~\cite{qu2025provably} & $1.0000 \pm 0.0000$ & $0.9550 \pm 0.0122$ & $0.9967 \pm 0.0024$ \\
& & BREW & $0.2883 \pm 0.0327$ & $0.0017 \pm 0.0024$ & $0.1558 \pm 0.0193$ \\
& & CORE-BREW & $0.3583 \pm 0.0357$ & $0.0100 \pm 0.0041$ & $0.2108 \pm 0.0204$ \\
\cmidrule(lr){2-6}
& \multirow{4}{*}{OpenGen}
& MPAC & $1.0000 \pm 0.0000$ & $0.8283 \pm 0.0170$ & $0.9947 \pm 0.0014$ \\
& & Qu et al.~\cite{qu2025provably} & $1.0000 \pm 0.0000$ & $0.9367 \pm 0.0201$ & $0.9967 \pm 0.0024$ \\
& & BREW & $0.2683 \pm 0.0155$ & $0.0033 \pm 0.0024$ & $0.1492 \pm 0.0103$ \\
& & CORE-BREW & $0.3550 \pm 0.0319$ & $0.0050 \pm 0.0000$ & $0.2058 \pm 0.0208$ \\

\bottomrule
\end{tabular}%
}
\end{table*}

Tables~\ref{tab:substitution_results}--\ref{tab:insertion_results} provide detailed detection and payload recovery results under 10\% substitution, deletion, and insertion attacks across OPT-1.3B and Mistral-7B on C4 and OpenGen. MPAC and Qu et al.~\cite{qu2025provably} generally retain high TPR and payload bit accuracy, but exhibit substantially higher FPR than BREW and CORE-BREW. In contrast, BREW and CORE-BREW maintain low observed FPR across the evaluated settings.

Under substitution, CORE-BREW generally improves TPR over BREW and consistently improves payload bit accuracy across models and datasets. Under deletion, both methods already achieve near-perfect TPR, while CORE-BREW provides consistently higher payload recovery. Insertion is the most challenging token-level attack, causing substantial degradation in both TPR and payload bit accuracy; nevertheless, CORE-BREW generally improves detection and consistently improves payload recovery over BREW while maintaining a mean observed FPR of at most $1\%$.

Overall, these results demonstrate that CORE-BREW yields the most pronounced performance gains under perturbations that severely disrupt the detector-facing bit sequence, whereas deletion-induced misalignment is effectively mitigated by the window-shifting mechanism alone.

\subsection{Detailed Results under Paraphrasing Attacks}
\label{Appendix:Paraphrasing_attacks}

\begin{table*}[t]
\centering
\caption{Detection and payload recovery performance under T5 and DIPPER paraphrasing attacks on OPT-1.3B with C4 and OpenGen at $T=500$. We report TPR, FPR, and payload bit accuracy as mean $\pm$ standard deviation over multiple runs.}
\label{tab:paraphrasing}
\small
\setlength{\tabcolsep}{6pt}
\renewcommand{\arraystretch}{1.05}

\begin{tabular}{l l l c c c}
\toprule
\textbf{Attack} & \textbf{Dataset} & \textbf{Method}
& \textbf{TPR} & \textbf{FPR} & \textbf{Payload Bit Acc.} \\
\midrule

\multirow{8}{*}{T5}
& \multirow{4}{*}{C4}
& MPAC & $1.0000 \pm 0.0000$ & $0.8250 \pm 0.0071$ & $0.9954 \pm 0.0003$ \\
& & Qu et al.~\cite{qu2025provably} & $1.0000 \pm 0.0000$ & $0.9400 \pm 0.0187$ & $1.0000 \pm 0.0000$ \\
& & BREW & $0.6850 \pm 0.0187$ & $0.0000 \pm 0.0000$ & $0.4683 \pm 0.0150$ \\
& & CORE-BREW & $0.7450 \pm 0.0122$ & $0.0083 \pm 0.0062$ & $0.5608 \pm 0.0326$ \\ 
\cmidrule(lr){2-6}
& \multirow{4}{*}{OpenGen}
& MPAC & $1.0000 \pm 0.0000$ & $0.8883 \pm 0.0085$ & $0.9934 \pm 0.0023$ \\
& & Qu et al.~\cite{qu2025provably} & $1.0000 \pm 0.0000$ & $0.9567 \pm 0.0024$ & $1.0000 \pm 0.0000$ \\
& & BREW & $0.7733 \pm 0.0332$ & $0.0000 \pm 0.0000$ & $0.5575 \pm 0.0159$ \\
& & CORE-BREW & $0.8067 \pm 0.0085$ & $0.0083 \pm 0.0024$ & $0.6233 \pm 0.0059$ \\

\midrule

\multirow{8}{*}{DIPPER}
& \multirow{4}{*}{C4}
& MPAC & $0.9967 \pm 0.0024$ & $0.8433 \pm 0.0062$ & $0.9750 \pm 0.0033$ \\
& & Qu et al.~\cite{qu2025provably} & $0.9967 \pm 0.0024$ & $0.9417 \pm 0.0062$ & $0.9633 \pm 0.0062$ \\
& & BREW & $0.4850 \pm 0.0108$ & $0.0033 \pm 0.0024$ & $0.3050 \pm 0.0134$ \\
& & CORE-BREW & $0.5667 \pm 0.0165$ & $0.0050 \pm 0.0041$ & $0.3967 \pm 0.0062$ \\
\cmidrule(lr){2-6}
& \multirow{4}{*}{OpenGen}
& MPAC & $0.9917 \pm 0.0062$ & $0.9100 \pm 0.0108$ & $0.9621 \pm 0.0049$ \\
& & Qu et al.~\cite{qu2025provably} & $0.9950 \pm 0.0041$ & $0.9267 \pm 0.0047$ & $0.9350 \pm 0.0122$ \\
& & BREW & $0.5150 \pm 0.0248$ & $0.0017 \pm 0.0024$ & $0.3422 \pm 0.0187$ \\
& & CORE-BREW & $0.5983 \pm 0.0118$ & $0.0067 \pm 0.0024$ & $0.4042 \pm 0.0171$ \\ 

\bottomrule
\end{tabular}
\end{table*}

Table~\ref{tab:paraphrasing} provides the detailed numerical results corresponding to the paraphrasing evaluation discussed in the main text. MPAC~\cite{yoo2024advancing} and Qu et al.~\cite{qu2025provably} retain near-perfect TPR and high payload bit accuracy under both T5 and DIPPER, but exhibit substantially higher FPR. In contrast, BREW~\cite{kim2026blockwise} and CORE-BREW maintain low observed FPR across both datasets.
Compared with BREW, CORE-BREW consistently improves both TPR and payload bit accuracy under T5 and DIPPER while keeping the observed FPR below $1\%$. Although paraphrasing, particularly DIPPER, remains challenging for payload recovery, these results support the robustness gains from calibrated soft information over the hard-decision BREW baseline.
\FloatBarrier

\section{Implementation Details and Reproducibility}
\label{app:impl_repro_details}


This appendix summarizes the implementation choices and settings required to reproduce our results, with special emphasis on: (i) the implementation of CORE-BREW and its hard-decision ablation (Sections~\ref{sec:method}--\ref{sec:experiments}), (ii) entropy-aware safeguards for low-entropy contexts, including erasure handling and bias control, (iii) detection thresholds, score-based rejection, and FPR--TPR evaluation, and (iv) the licenses and terms of use for external assets used in our experiments. Our implementation builds on standard open-source tooling for deep learning and text processing~\cite{Paszke2019,wolf2020transformers}.

\subsection{Codebase Organization}

The codebase is organized as a Python package. The implementation follows the block-wise designated-codeword framework underlying BREW~\cite{kim2026blockwise}, together with ECC-based components shared with related multi-bit watermarking methods~\cite{qu2025provably}. The package includes modules for:
(i) fixed keyed vocabulary partitioning and block-dependent payload-to-codeword construction,
(ii) embedding, including logit interception and Constant Hit-Rate biasing,
(iii) extraction and detection, including bit extraction, window shifting, decoding, and scoring,
(iv) attack pipelines for token edits and paraphrasing,
and (v) evaluation, including metric computation and result aggregation.
Entry-point scripts provide executable examples for running generation and detection and producing the corresponding output summaries.
The accompanying implementation is intended to support reproducibility of the main CORE-BREW pipeline, while full-scale table and figure generation relies on the saved artifacts and evaluation logs described below.

To support reproducibility and auditing, we save intermediate artifacts, including generated texts (watermarked and unwatermarked), attacked texts, per-block match decisions, and, when applicable, per-block scores. This allows evaluation and plotting to be rerun without regenerating text.

\subsection{Software and hardware environment}
\label{app:hardware}
\paragraph{Software stack.}
All experiments are implemented in Python and use: PyTorch for model execution~\cite{Paszke2019}, Hugging Face Transformers for model loading and generation~\cite{wolf2020transformers}, and SentencePiece for models requiring subword tokenization~\cite{Kudo2018}. We use standard scientific Python utilities (NumPy/SciPy) for evaluation and statistics.

\paragraph{Hardware.}
Generation is run on CUDA-enabled GPUs when available. Detection can be run on CPU for offline or batched evaluation, although model-aware erasure computation requires access to the generator or an equivalent scoring model.

\paragraph{Compute resources.}
All experiments are inference-time evaluations; no model training or fine-tuning is performed. Experiments were run on a local CUDA 13.0 server with four NVIDIA RTX PRO 6000 Blackwell GPUs, each with 96GB of VRAM, an Intel Xeon 6740P CPU with 96 physical cores, and 1.5TB of RAM.
To characterize inference-time efficiency, we measure per-prompt generation latency, throughput, generation overhead relative to unwatermarked decoding, and detection latency. Table~\ref{tab:compute_efficiency} shows that CORE-BREW requires $8.65$ s per prompt, corresponding to $57.8$ tokens/s and an $11.3\%$ generation overhead. Its hard-decision ablation has nearly identical generation cost ($8.77$ s, $57.0$ tokens/s, and $12.8\%$ overhead), and both are substantially more efficient during generation than MPAC, Qu et al.~\cite{qu2025provably}, and BREW.
The additional detection cost of CORE-BREW includes model-aware probability recomputation, LLR-based scoring, and Chase-style list decoding: CORE-BREW requires $856$ ms per prompt on average, compared with $491$ ms for the hard-decision ablation. Nevertheless, CORE-BREW completes detection in under one second on average and remains substantially faster than Qu et al.~\cite{qu2025provably}, which requires approximately $27.2$ s per prompt. Overall, CORE-BREW introduces modest generation overhead while trading additional detection cost for model-aware soft-decision processing.

\begin{table*}[t]
\centering
\caption{Inference-time computational efficiency on OPT-1.3B with C4. Generation time is reported per prompt, generation overhead is measured relative to unwatermarked decoding, and detection latency is reported in milliseconds.}
\label{tab:compute_efficiency}
\small
\setlength{\tabcolsep}{5pt}
\renewcommand{\arraystretch}{1.05}

\resizebox{\textwidth}{!}{%
\begin{tabular}{l c c c c}
\toprule
\textbf{Method}
& \textbf{Gen. time (s)}
& \textbf{Tok./s}
& \textbf{Gen. overhead}
& \textbf{Detect. (ms)} \\
\midrule
Unwatermarked & $7.77 \pm 0.22$ & $64.4$ & -- & -- \\
MPAC ($\delta=3$) & $59.19 \pm 2.11$ & $8.4$ & $+661.6\%$ & $631 \pm 105$ \\
Qu et al.~\cite{qu2025provably} ($\delta=3$) & $21.11 \pm 2.67$ & $24.0$ & $+171.6\%$ & $27154 \pm 8091$ \\
BREW ($\delta=3$) & $11.24 \pm 0.19$ & $44.5$ & $+44.7\%$ & $511 \pm 77$ \\
CORE-BREW w/o Soft Decoding & $8.77 \pm 0.14$ & $57.0$ & $+12.8\%$ & $\mathbf{491 \pm 125}$ \\
CORE-BREW & $\mathbf{8.65 \pm 0.15}$ & $\mathbf{57.8}$ & $\mathbf{+11.3\%}$ & $856 \pm 414$ \\
\bottomrule
\end{tabular}%
}
\end{table*}

\paragraph{Reproducibility and Determinism.}
Exact bitwise determinism can be affected by nondeterministic GPU kernels and external paraphrasers. We therefore report averages over multiple seeds and log the seeds, configurations, generated outputs, and detection summaries needed to audit the reported results.

\subsection{Models, Tokenization, and Datasets}

\paragraph{Models and checkpoints.}
We evaluate publicly available autoregressive LLM checkpoints loaded via Transformers~\cite{wolf2020transformers}.
Model and tokenizer identifiers are recorded in each run's configuration.

\paragraph{Tokenization.}
All watermark operations, including partition membership, bit extraction, and block reconstruction, are performed on token IDs.
Word-level attacks such as synonym substitution are filtered according to the token-level constraints required by each experiment.

\paragraph{Datasets and prompt sampling.}
We use C4 as the long-form corpus and OpenGen, consisting of short two-sentence prompts derived from WikiText-103, as the short-form prompt set. Prompt splits are kept disjoint across development and test evaluation. All reported numbers are computed on test prompts, and the same splits are reused across schemes and attacks for fair comparison.

\subsection{Asset Licenses and Terms of Use}
\label{app:asset-licenses}

Table~\ref{tab:asset-licenses} summarizes the external assets used in our experiments. We use all assets under their stated licenses or terms of use and do not redistribute pretrained model checkpoints or third-party baseline code. The supplementary package includes small processed prompt/evaluation subsets derived from existing public datasets; the original dataset licenses and terms of use apply to these derived files.

\begin{table}[t]
    \centering
    \caption{External assets used in our experiments. We list dataset and model
    assets separately because the corresponding licenses and terms differ by
    asset. Detailed URLs, revisions, commit hashes, and package versions are
    provided in the supplementary README.}
    \label{tab:asset-licenses}
    \setlength{\tabcolsep}{4pt}
    \resizebox{\columnwidth}{!}{%
    \begin{tabular}{p{0.20\columnwidth} p{0.20\columnwidth} p{0.26\columnwidth} p{0.26\columnwidth}}
    \toprule
    \textbf{Asset} & \textbf{Role} & \textbf{License / terms} & \textbf{Use in this work} \\
    \midrule
    OPT-1.3B
    & Target LLM
    & OPT license / model terms
    & Used for inference-time evaluation; checkpoint not redistributed. \\
    \midrule
    Mistral-7B-v0.3
    & Target LLM
    & Apache-2.0
    & Used for inference-time evaluation; checkpoint not redistributed. \\
    \midrule
    C4 processed subset
    & Prompt/evaluation corpus
    & ODC-BY 1.0; Common Crawl terms
    & Processed prompt/evaluation subset included in the supplementary package; original terms apply. \\
    \midrule
    OpenGen processed subset
    & Derived prompt/evaluation corpus
    & Terms inherited from the source corpus
    & Derived from WikiText-103 for prompt-based evaluation; applicable source terms apply. \\
    \midrule
    WikiText-103
    & Source corpus
    & CC BY-SA / GFDL according to the source release
    & Used to construct the derived OpenGen prompt/evaluation subset. \\
    \midrule
    MarkLLM
    & Evaluation framework
    & Apache-2.0
    & Used as a unified watermarking evaluation framework with attribution. \\
    \midrule
    MPAC implementation
    & Baseline
    & Stated repository license
    & Used for research comparison; original notices preserved where applicable. \\
    \midrule
    Qu et al. implementation
    & Baseline
    & Stated repository license
    & Used for research comparison; original notices preserved where applicable. \\
    \midrule
    BREW implementation
    & Baseline
    & Stated repository license
    & Used for comparison; original notices preserved where applicable. \\
    \midrule
    PyTorch; NumPy; SciPy
    & Software libraries
    & BSD-style licenses
    & Used for model inference, numerical computation, and evaluation. \\
    \midrule
    Transformers; SentencePiece
    & Software libraries
    & Apache-2.0
    & Used for loading pretrained models and tokenizers. \\
    \bottomrule
    \end{tabular}%
    }
\end{table}

\subsection{Generation and watermark embedding configuration}
\paragraph{Sampling settings.}
Unless otherwise stated, generation uses temperature $1.0$ and nucleus (top-$p$) sampling with $p=0.9$~\cite{holtzman2020curious}. We log all decoding parameters (temperature, top-$p$, max new tokens, stop criteria). Constant Hit-Rate calibration targets the biased pre-top-$p$ distribution, while the realized hit rate after nucleus sampling may differ slightly.

\paragraph{ECC and block parameters.}
We use a binary BCH code $C\subseteq\{0,1\}^n$ with code parameters $(n,k,t)$ as specified in Section~\ref{sec:experiments} and the appendix ablations~\cite{lin2004error,richardson2008modern}. The block length is set to $n$. The payload $m\in\{0,1\}^k$ is either fixed per run or sampled per text,
and the mapping from payload to per-block designated codewords $c^{(j)}$ is deterministic given the key and block index.

\paragraph{Constant Hit-Rate embedding.}
We implement the Constant Hit-Rate embedding rule defined in Section~\ref{sec:method}. At each generation step, the base target-list mass is computed from the model logits, with $m_t$ clamped to $[\varepsilon,1-\varepsilon]$ for numerical stability. The resulting bias is then subject to the distortion guards described below.

\paragraph{Distortion guards and erasures.}
We apply the distortion-aware safeguards defined in Section~\ref{sec:method} and implemented in Appendix~\ref{app:algorithms}, Algorithms~\ref{alg:chr_embedding} and~\ref{alg:detect}. The raw bias is clipped to $[-\delta_{\max},\delta_{\max}]$, and positions that fail the peakness, entropy, or post-bias feasibility checks are left unbiased and marked as erasures. Erased positions contribute zero LLR during model-aware detection. We log the erasure rate and tail statistics of $|\delta_t|$ for diagnostic analysis.

\subsection{Detection Implementation}
\label{app:detection-implementation}

Detection follows BREW's block-wise designated-codeword framework with window shifting~\cite{kim2026blockwise}. For each block, we test candidate offsets $s\in\{-s_{\max},\ldots,s_{\max}\}$ around the nominal block anchor and construct the corresponding hard-decision or LLR representation (Appendix~\ref{app:algorithms}, Algorithm~\ref{alg:detect}). CORE-BREW uses calibrated LLR evidence with score-based rejection and Chase-style list decoding. To evaluate the contribution of soft-decision processing, we additionally consider a hard-decision ablation that uses the same Constant Hit-Rate embedding and erasure decisions, but fills erased positions with key-derived PRF bits and applies bounded-distance designated-codeword verification.

\paragraph{Per-bit LLRs.}
Under Constant Hit-Rate embedding, each non-erased position has constant LLR magnitude
\[
\lambda=\log\!\left(\frac{p^\star}{1-p^\star}\right),
\]
with sign determined by keyed list membership. Erased positions contribute zero LLR.

\paragraph{Hard-decision ablation.}
The hard-decision ablation retains bounded-distance designated-codeword verification on the PRF-filled hard-decision representation. For each candidate offset, erased positions are filled with key-derived PRF bits, and the resulting hard-decision vector is evaluated using a bounded-distance BCH decoder. A block is accepted only when bounded-distance BCH decoding returns the designated codeword. This ablation does not use likelihood-based scoring or list decoding and provides a hard-decision reference for evaluating the contribution of soft-decision processing under the same Constant Hit-Rate embedding and erasure decisions.


\paragraph{CORE-BREW soft-decision detection.}
CORE-BREW preserves erasures as zero-LLR positions and combines calibrated LLR scoring, designated-codeword verification, and Chase-style list decoding. For each candidate offset, we compute the designated-codeword score
\[
S^{(j)}(s)=\sum_{b=0}^{n-1}(2c^{(j)}[b]-1)\Lambda_b^{(j,s)},
\]
and require the candidate score to exceed a block-level threshold $\tau_{\mathrm{blk}}$. The Chase-style list-decoding wrapper explores a bounded set of candidate hard decisions obtained from the least-reliable positions, applies bounded-distance decoding to each candidate, and evaluates the resulting codewords using the LLR-consistent score~\cite{chase1972,lin2004error,richardson2008modern}. A block is accepted only when a decoded candidate equals the designated codeword and satisfies the score threshold. This procedure can recover the designated codeword beyond the bounded-distance radius, although beyond-radius recovery is not the sole source of detection gain; erasure-aware scoring and LLR-consistent candidate evaluation can also improve discrimination. The score threshold therefore serves as an additional filtering criterion within the overall acceptance rule.

\paragraph{Detection thresholds and list-decoding parameters.}
For text-level detection, let $M_{\mathrm{match}}$ denote the number of accepted blocks among $M$ complete blocks. A sample is classified as watermarked when
\[
\frac{M_{\mathrm{match}}}{M}\ge\theta.
\]
For CORE-BREW, a shifted candidate block is accepted only if the Chase-decoded codeword equals the designated codeword and its LLR score exceeds the block-level threshold $\tau_{\mathrm{blk}}$. Unless otherwise specified, we set $\tau_{\mathrm{blk}}=5.0$. For Chase-style list decoding, we sort bit positions by increasing reliability $|\Lambda_b|$, select the $\ell$ least reliable positions as flip candidates, and enumerate flip patterns up to a maximum list size $L$. We use $\ell=2$ and $L=32$, resulting in at most $\min(2^\ell,L)=4$ candidate hard-decision vectors per shifted block. For insertion and deletion attacks, we additionally search over $2s_{\max}+1$ shifted alignments per block.

\paragraph{Detector capability setting.}
When model-aware erasures are enabled, the detector reapplies the same erasure rule to the observed prefix using the base model or an equivalent scoring model, rather than assuming access to the original generation-time erasure mask. Because the observed prefix may differ from the generation-time prefix after edits, the resulting detection-time erasure decisions need not coincide with those made during embedding. This setting is most appropriate for provider-side verification. For token-only detection, erasures are disabled, and the detector uses the constant-magnitude LLR rule without prefix-probability recomputation. We distinguish these detector-access settings explicitly in the corresponding experiments.

\subsection{Threshold Selection and Reporting}

\paragraph{Reporting protocol.}
We report detection performance using text-level TPR and FPR. Unless otherwise stated, each method is evaluated using its default or recommended detection threshold specified in the corresponding experiment. When threshold sensitivity is evaluated, we report TPR and FPR at each tested operating point.

\paragraph{Variability across runs.}
Unless otherwise stated, we report mean $\pm$ standard deviation over multiple random seeds. The reported variability reflects randomness from prompt sampling, watermark key/payload sampling, model generation, and attack procedures.

\paragraph{Threshold sweeps.}
For hard-decision methods such as BREW and the hard-decision ablation, threshold sweeps vary the text-level match threshold $\theta$. For CORE-BREW, sweeps may vary the block-level score threshold $\tau_{\mathrm{blk}}$ and the text-level threshold $\theta$. Unwatermarked texts are used to estimate the corresponding FPR.

\paragraph{FPR Evaluation under Perturbed Null Distribution ($H_0$).}
To assess whether false positive behavior remains stable under distribution shift, we apply the same attack pipelines to unwatermarked texts and report the resulting attacked-$H_0$ FPR. This is particularly important for score-based detectors, whose score distributions may shift under paraphrasing or token-level edits.

\subsection{Attack pipelines and evaluation metrics}
\paragraph{Token-level edits.}
We implement substitution/insertion/deletion pipelines inspired by TextAttack~\cite{morris-etal-2020-textattack}.
For token-preserving substitution, we restrict to synonym replacements that preserve tokenization length to maintain alignment.
For insertion/deletion-like attacks, we allow length changes and rely on window shifting during detection.

\paragraph{Paraphrasing.}
Paraphrasing attacks use a separate paraphraser model (distinct from the generator) following common paraphrase evaluation setups~\cite{Wieting2018,krishna2023paraphrasing}. The paraphraser does not receive the watermark key. We quantify paraphrase strength using BLEU~\cite{Papineni2002} and semantic similarity via BERTScore (or a comparable embedding-based metric)~\cite{zhang2020bertscore}.

\paragraph{Metrics.}
We report text-level TPR and FPR for watermark detection, payload bit accuracy for payload recovery, and PPL, BLEU, and BERTScore for text quality.



\subsection{Random Seeds, Logging, and Released Artifacts}

We control and record randomness across prompt sampling, watermark key generation, payload sampling, model sampling RNG states (Python/NumPy/PyTorch CPU and CUDA), and attack procedures. We log all hyperparameters in machine-readable configuration files (YAML/JSON), including the model checkpoint, dataset split, $(n,k,t)$, $p^\star$, $s_{\max}$, $\theta$, $\tau_{\mathrm{blk}}$, $\delta_{\max}$, $q_{\max}$, list-decoding parameters, and generation/attack settings.
We release (or will release upon acceptance, consistent with policy) scripts to: (i) regenerate watermarked and unwatermarked corpora from prompts, (ii) apply attack pipelines, (iii) run detection and threshold sweeps, and (iv) reproduce the reported figures and tables from saved artifacts.

\section{Limitations, Ethics, and Threat Model Clarification}
\label{app:limitations}

This appendix clarifies the threat model and limitations of our approach, and discusses ethical considerations.
We clarify three aspects of the proposed framework:
(i) the distinction between CORE-BREW's soft-decision detector and its hard-decision ablation,
(ii) the practical risks of distortion in low-entropy contexts and our safeguards,
and (iii) the limits of robustness under strong semantic rewriting.

\subsection{Threat model and deployment assumptions}
\paragraph{Keyed provenance, not universal detection.}
Our objective is \emph{keyed} provenance verification: given a secret key $K$, a provider (or authorized verifier) can test whether a text was generated by a watermarked model under that key. This is distinct from open-set, keyless classifiers for ``AI-generated'' text and does not attempt to detect text generated by arbitrary models without key access~\cite{kirchenbauer2023watermark,mitchell2023detectgpt}.


\paragraph{Detector capabilities: token-only vs.\ model-aware.}
We distinguish two detector capability settings for designated-codeword verification. In both settings, verification uses the observed text, the secret key $K$, and the designated codewords $\{c^{(j)}\}$ associated with the payload under test. The distinction concerns access to model probabilities.
\begin{itemize}
\item \textbf{Token-only setting:} in principle, observed tokens can be mapped to hard bits using the fixed keyed partition and evaluated using the constant-magnitude LLR rule without model access, with model-dependent erasure checks disabled.
\item \textbf{Model-aware setting:} the detector additionally uses the base model (or an equivalent scoring model) to recompute token probabilities on the observed prefix and apply the entropy-aware erasure rule (Sections~\ref{sec:method}--\ref{sec:experiments} and Appendix~\ref{app:impl_repro_details}).
\end{itemize}
Our erasure-aware CORE-BREW experiments use the model-aware setting, which is practical for provider-side provenance verification.

\paragraph{Adversary goals and knowledge.}
We assume an adversary can post-process the text via paraphrasing, synonym substitution, insertion, deletion, or reformatting, as commonly studied in watermark robustness evaluations~\cite{morris-etal-2020-textattack,wolff2020attacking}. We primarily consider \emph{black-box} adversaries who do not know the key $K$. A fully informed attacker who knows $K$ and can exactly simulate the detector could, in principle, remove the watermark by targeted rewriting; like most keyed watermarking schemes, we do not claim robustness against such a worst-case white-box adversary.

\subsection{Limitations of the channel model}

\paragraph{Idealized BSC and independence assumptions.}
Our theory models fully calibrated, non-erased positions under Constant Hit-Rate embedding as approximately following a position-homogeneous binary symmetric channel with parameter $p^\star$. While channel shaping reduces context dependence relative to fixed-bias watermarking, real LLM generation is not strictly i.i.d.: token distributions depend on long-range context, and adversarial edits introduce correlated changes. Consequently, the BSC model and approximate block independence are best viewed as \emph{analysis tools} rather than exact realities. We therefore complement theoretical bounds with empirical evaluations of the observed FPR--TPR trade-off in Section~\ref{sec:experiments}. For deployment, detection thresholds should be validated using in-domain unwatermarked text.


\paragraph{Fixed partitions and hard-decision interpretation.}
Our evaluated instantiation uses a fixed keyed vocabulary partition so that a token's partition bit remains unchanged when its position shifts under insertion/deletion attacks. However, reuse of the same partition can introduce dependencies through repeated tokens and natural-language correlations, so aggregate FPR bounds should be interpreted as idealized analysis tools complemented by empirical FPR evaluation on unwatermarked text. Similarly, the hard-decision ablation applies bounded-distance verification to the detector's PRF-filled hard-decision representation, which may differ from the raw token-derived hard vector when erased positions are replaced.

\paragraph{Shift search limitations.}
Window shifting partially addresses insertion/deletion noise by searching a bounded set of alignments.
However, large or highly structured edits can break local alignment assumptions and reduce match rates.
More powerful synchronization mechanisms (e.g., learned alignment or marker-based synchronization) are outside our scope.

\subsection{Limitations of Robustness under Semantic Rewriting}

Paraphrasing and semantic rewriting can eventually destroy most token-level watermarks, including ours, if the attacker is allowed sufficient edit budget and can rewrite aggressively~\cite{wolff2020attacking,krishna2023paraphrasing}. Our goal is to improve robustness in the practically relevant regime of \emph{moderate} edits while maintaining low false-positive rates. We evaluate the resulting FPR--TPR behavior under paraphrasing attacks to characterize degradation under semantic rewriting (Section~\ref{sec:experiments}). We do not claim unconditional robustness to arbitrary meaning-preserving transformations, which remains an open challenge for token-level watermarking.

\subsection{Low-Entropy Contexts and Quality Risks}

Enforcing a fixed hit-rate in low-entropy contexts can require large logit shifts when the base model assigns very small probability mass to the target list. Such shifts may affect fluency, factuality, or style~\cite{holtzman2020curious}. To mitigate this risk, we use bias control and feasibility-based erasures (Section~\ref{sec:method}), and analyze their effects through text-quality and erasure ablations in Appendices~\ref{app:text-quality-sweep} and~\ref{app:erasure-ablation}. These safeguards trade off watermark evidence for reduced distortion: when many positions are erased, the effective evidence per block decreases, so longer text or more blocks may be needed for reliable detection.

\subsection{Acceptance Regions under Hard- and Soft-Decision Decoding}

A persistent technical ambiguity in soft-decision watermarking is whether one can simultaneously
(i) preserve the hard-decision bounded-distance acceptance region exactly, thereby retaining the corresponding combinatorial false-positive analysis, and
(ii) exploit soft LLR information to improve decoding beyond what is available from the raw hard decision.
CORE-BREW separates these roles through its hard-decision ablation and soft-decision detector.

The \textbf{CORE-BREW w/o Soft Decoding} ablation retains the same Constant Hit-Rate embedding and erasure decisions as CORE-BREW, but applies bounded-distance designated-codeword verification to the PRF-filled hard-decision representation. Its acceptance region is therefore restricted to the bounded-distance decoding neighborhood of the designated codeword under the corresponding idealized analysis assumptions.

In contrast, \textbf{CORE-BREW} uses calibrated LLR scoring and limited Chase-style list decoding~\cite{chase1972,lin2004error}. Its acceptance rule is not restricted to the radius-$t$ neighborhood of the raw hard-decision vector: a bounded set of candidate hard decisions is explored, and a block is accepted only when a decoded candidate equals the designated codeword and its score exceeds the block-level threshold $\tau_{\mathrm{blk}}$. Consequently, the soft-decision acceptance rule may extend beyond the hard-decision bounded-distance region and is analyzed separately using score-based bounds together with empirical FPR evaluation. The corresponding analysis is provided in Section~\ref{sec:Theory}, with empirical false-positive behavior reported in Section~\ref{sec:experiments}.

This distinction provides a matched hard-decision reference for evaluating the contribution of soft-decision processing without treating the ablation as a separate CORE-BREW configuration.

\subsection{Ethical considerations and potential misuse}
\label{app:Ethical}
\paragraph{Intended use.}
Our primary intended use is provenance verification for transparency and accountability, including attribution of model-generated content, auditing of policy compliance tags, and tracing misuse when a model provider is responsible for outputs. The designated-codeword design targets extremely low FPR to reduce the risk of falsely accusing a human author.

\paragraph{Misuse and over-reliance.}
Watermarks should not be treated as a sole or definitive indicator of authorship. Watermarks can be removed or corrupted, and absence of a watermark does not imply human authorship. Over-reliance can lead to erroneous conclusions, especially under distribution shift or when the detector is used outside its validated domain. We therefore recommend using watermark signals as one component in a broader provenance and safety pipeline.

\paragraph{Privacy and data handling.}
Multi-bit payloads can encode identifiers. Even when keys are secret, embedding persistent identifiers may raise privacy concerns. We recommend minimizing payload content, adopting rotation/expiration policies for identifiers, and restricting detector access. We also emphasize that our evaluation does not require training on private user data; it uses public corpora and model outputs.


\end{document}